\documentclass[aps,pra,superscriptaddress,10pt,tightenlines,twocolumn,nofootinbib]{revtex4}
\usepackage{amsmath,amsthm}
\usepackage{amssymb}
\usepackage[mathscr]{eucal}
\usepackage{graphicx}
\usepackage{lscape}
\usepackage{url}
\usepackage{hyperref}
\usepackage{xifthen} 
\hypersetup{colorlinks=true,urlcolor=blue}
\newtheorem{assump}{Assumption}
\newtheorem{theorem}{Theorem}
\newtheorem{corollary}[theorem]{Corollary}
\newtheorem{lemma}[theorem]{Lemma}

\theoremstyle{remark}
\newtheorem*{remark}{Remark}
\newtheorem{definition}{Definition}

\newcommand{\be}{\begin{equation}}
\newcommand{\ee}{\end{equation}}
\newcommand{\ea}[1]{\begin{align}#1\end{align}}
\newcommand{\bc}{\begin{center}}
\newcommand{\ec}{\end{center}}
\newcommand{\bmt}{\begin{pmatrix}}
\newcommand{\emt}{\end{pmatrix}}
\newcommand{\fd}[1]{\mathbb{#1}}

\newcommand{\bsmt}{\left(\begin{smallmatrix}}
\newcommand{\esmt}{\end{smallmatrix}\right)}

\newcommand{\vpu}[1]{^{\vphantom{#1}}}

\newcommand{\la}{\langle}
\newcommand{\ra}{\rangle}

\newcommand{\bs}{\Delta_{\rm{e}}}
\newcommand{\ps}{\Delta}
\newcommand{\dl}[1]{#1^{*}}
\newcommand{\ddl}[1]{#1^{**}}
\newcommand{\bo}{B_{\rm{o}}}
\newcommand{\bi}{B_{\rm{i}}}
\newcommand{\midball}{B_{\rm m}}
\newcommand{\midsphere}{S_{\rm m}}

\newcommand{\so}{S_{\rm{o}}}
\newcommand{\si}{S_{\rm{i}}}
\newcommand{\ro}{r_{\rm{o}}}
\newcommand{\ri}{r_{\rm{i}}}
\newcommand{\rmid}{r_{\rm m}}
\newcommand{\allqplexes}{\mathcal{Q}}
\newcommand{\allgerms}{\mathcal{A}}
\newcommand{\stem}{stem}
\newcommand{\envelope}{envelope}

\newcommand{\qp}{Q}

\newcommand{\pr}{G}
\newcommand{\gp}{\mathcal{G}}
\newcommand{\qplex}{qplex}
\newcommand{\qplexes}{qplexes}
\newcommand{\Qplexes}{Qplexes}

\newcommand{\qgp}[1]{G_{#1}}

\newcommand{\typ}[1]{\mathcal{T}_{#1}}
\newcommand{\sym}[1]{\mathcal{S}_{#1}}
\newcommand{\ort}[1]{\mathcal{O}_{#1}}
\newcommand{\herm}{B_H}
\newcommand{\dens}{S}
\newcommand{\cP}{\mathcal{P}}
\newcommand{\cR}{\mathcal{R}}
\newcommand{\inprod}[2]{{\left\langle #1 , #2 \right\rangle}}

\newcommand{\ket}[1]{{\left| #1 \right \rangle}}

\DeclareMathOperator{\Ot}{O}

\DeclareMathOperator{\PU}{PU}
\DeclareMathOperator{\EU}{EU}
\DeclareMathOperator{\PEU}{PEU}
\DeclareMathOperator{\M}{M}
\newcommand{\stm}{\mathscr{S}}
\newcommand{\env}{\mathscr{E}}

\DeclareMathOperator{\cc}{cc}

\DeclareMathOperator{\Tr}{tr}

\DeclareMathOperator{\su}{\mathfrak{su}}

\newcommand{\arxiv}[2][]{\ifthenelse{\isempty{#1}}{\href{http://arxiv.org/abs/#2}{{\tt arXiv:\allowbreak{}#2}}} {\href{http://arxiv.org/abs/#2}{{\tt arXiv:\allowbreak{}#2 [#1]}}}}

\newcommand{\booktitle}{\textsl}
\newcommand{\hrefdoi}[2]{\href{https://dx.doi.org/#1}{#2}}

\begin{document}

\title{Introducing the Qplex: A Novel Arena for Quantum Theory}

\author{Marcus Appleby}
\affiliation{Centre for Engineered Quantum Systems, School of Physics, University of Sydney, Sydney, Australia}
\author{Christopher A.\ Fuchs}
\affiliation{Physics Department, University of Massachusetts Boston, Boston, MA 02125, USA}
\affiliation{Max Planck Institute for Quantum Optics, 85748 Garching, Germany}
\author{Blake C.\ Stacey}
\affiliation{Department of Physics, University of Massachusetts Boston, Boston, MA 02125, USA}

\author{Huangjun Zhu}

\affiliation{Institute for Theoretical Physics, University of Cologne, 50937 Cologne, Germany}

\date{\today}

\begin{abstract}
We reconstruct quantum theory starting from the premise that, as Asher Peres remarked, ``Unperformed experiments have no results.''  The tools of quantum information theory, and in particular the symmetric informationally complete (SIC) measurements, provide a concise expression of how exactly Peres's dictum holds true.  That expression is a constraint on how the probability distributions for outcomes of different, hypothetical and mutually exclusive experiments ought to mesh together, a type of constraint not foreseen in classical thinking.  Taking this as our foundational principle, we show how to reconstruct the formalism of quantum theory in finite-dimensional Hilbert spaces.  Along the way, we derive a condition for the existence of a $d$-dimensional SIC.
\end{abstract}

\maketitle

\setcounter{assump}{-1}

\section{Introduction}

%

The arena for standard probability theory is the probability simplex---that is, for a trial of $n$ possible outcomes, the continuous set $\Delta_n$ of all $n$-vectors ${p}\,$ with nonnegative entries $p(i)$ satisfying $\sum_i p(i)=1$.  But what is the arena for quantum theory?  The answer to this question depends upon how one views quantum theory.  If, for instance, one views it as a noncommutative {\it generalization\/} of probability theory, then the arena could be the convex sets of density operators and positive-operator-valued measures over a complex Hilbert space. In contrast, Refs.\ \cite{fuchs2013,fuchs2009,Voldemort} have argued that quantum theory is not so much a generalization of probability theory as an {\it addition\/} to it.  This means that standard probability theory is never invalidated, but that further rules must be added to it when the subject matter concerns measurements on quantum systems.  One implication of this is that behind every application of quantum theory is a more basic simplex, which through a not-yet-completely-understood consistency requirement, gets trimmed or cropped to a convex subset isomorphic to the usual space of quantum states~\cite{fuchs2002}.\footnote{See also \cite[p.\ 487]{fuchs2010} for the historical roots of this idea.} In the specific context formalized below, we call an arena of this sort---a suitably cropped simplex as the starting point for a full-fledged derivation of quantum theory---a {\it qplex}.  In a slogan: If the simplex is the starting point for probability theory, the qplex is the starting point for the quantum.

The introduction of a more basic simplex surrounding the qplex, however, should not be construed as a capitulation to the idea of a hidden-variable theory.  Rather it is an attempt to bring to the front of the formalism a foundational idea nicely captured by Asher Peres's famous quip ``unperformed experiments have no outcomes''~\cite{Peres78}.  Here the simplex stands for the outcomes of an experiment that will never be done, but could have been done.  How is probability theory all by itself to connect the one experiment to the other?  It has no tools for it.  But quantum theory does, through the Born rule, when suitably rewritten in the language of the qplex.  From this point of view, the meaning of the Born rule for probabilities in any actual experiment is that ``behind'' the experiment is a different, hypothetical experiment whose probabilities {\it must be taken into account\/} in the calculation.

To be concrete, let us rewrite quantum theory in a language that would make this apparent {\it were the right mathematical tool available}.  Consider the setting of a finite $d$-level quantum system, and suppose that one of the elusive symmetric informationally complete quantum measurements~\cite{zauner1999,renes2004} exists for it.  We shall call such an object a ``SIC'' for short.  A SIC is a set of $d^2$ rank-one projection operators $\Pi_i=|\psi_i\rangle\langle\psi_i|$ such that
\begin{equation}
\Tr (\Pi_k \Pi_l) = \frac{d\delta_{kl} + 1}{d + 1}\;.
\label{eq:SICOverlaps}
\end{equation}
For such a set of operators, one can prove that if they exist at all, they must be linearly independent, and rescaling each to $\frac{1}{d}\Pi_i$, they collectively give an informationally complete positive-operator-valued measure (POVM), i.e., $\sum_i\frac{1}{d}\Pi_i=I$.  Thus, for any quantum state $\rho$, a SIC can be used to specify a measurement for which the probabilities of outcomes $p(i)$ specify $\rho$ itself.  That is, if
\be
p(i) = \frac{1}{d} \Tr(\rho \Pi_i)\;,
\label{eq:SicProbabilities}
\ee
then
\begin{eqnarray}
\rho
 &=& \sum_{i=1}^{d^2} \left[(d+1) p(i) - \frac{1}{d}\right]\! \Pi_i \label{eq:rhoTermsProbs}\\
 &=& (d + 1) \sum_{i=1}^{d^2} p(i) \Pi_i - I.
\end{eqnarray}

Is it always possible to write a quantum state like this?\footnote{To
  our knowledge the first person to write down this expression was the
  Cornell University undergraduate Gabriel G.\ Plunk in an attachment
  to a 18 June 2002 email to one of us (CAF), though it went
  undiscovered for many years. See
  Ref.~\cite[pp.\ 472--474]{Fuchs2014}.}  Unfortunately, to date,
analytic proofs of SIC existence have only been found in dimensions
2--21, 24, 28, 30, 31, 35, 37, 39, 43 and 48~\cite{appleby2016,
  applebyprivate}.  However, very high-precision numerical
approximations (many to 8,000 and 16,000 digits) have been discovered
for all dimensions 2 to 147 without exception, plus some dimensions
sporadically beyond that---168, 172, 195, 199, 228, 259,
323, at last count~\cite{scott2010, scottprivate, hoangprivate}.  In
general, the mood of the community is that a SIC should exist in every
finite dimension $d$, but we call the SICs ``elusive'' because in more
than 18 years of effort no one has ever proven it.  See
Ref.~\cite{Galois} for an extensive bibliography on the subject.  For
the purpose of the present discussion, let us suppose that at least
one SIC can be found in any finite dimension $d$.

One can now see how to express quantum-state space as a proper subset $\qp$ of a probability simplex $\Delta_{d^2}$ over $d^2$ outcomes.  That it cannot be the full simplex comes about from the following consideration:  For any ${p}\in \Delta_{d^2}$, Eq.~(\ref{eq:rhoTermsProbs}) gives a Hermitian operator $\rho$ with trace~1, but the operator may not be positive-semidefinite as is required of a density operator.  Instead, the density operators correspond to a convex subset specified by its extreme points, the pure states $\rho^2 = \rho$.  Thanks to an observation by Jones, Flammia and Linden~\cite{flammia2004,jones2005}, we can also characterize pure states as those Hermitian matrices satisfying
\begin{equation}
\Tr\rho^2 = \Tr\rho^3 = 1.
\end{equation}
This expression of purity yields two conditions on the probability
distributions ${p}$ \cite{fuchs2009,appleby2011,fuchs2013}.  First,
\begin{equation}
\sum_{i=1}^{d^2} p(i)^2 = \frac{2}{d(d+1)},
\label{eq:purity1}
\end{equation}
and second,
\begin{equation}
\sum_{ijk} c_{ijk}\, p(i) p(j) p(k) = \frac{d+7}{(d+1)^3},
\label{eq:purity2}
\end{equation}
where we have defined the real-valued, completely symmetric three-index tensor
\begin{equation}
c_{ijk} = \hbox{Re}\, \Tr (\Pi_i \Pi_j \Pi_k).
\label{eq:c-tensor}
\end{equation}
The full state space $\qp$ is the convex hull of probability distributions satisfying Eqs.~(\ref{eq:purity1}) and (\ref{eq:purity2}).

So the claim can be made true, but what a strange-looking set the quantum states become when written in these terms!  What could account for it except already knowing the full-blown quantum theory as usually formulated?

Nevertheless, every familiar operation in the textbook quantum formalism has its translation into the language of this underlying probability simplex, properly restricted to the subset $\qp$.  For example, given a quantum state $\rho$, one uses the Born rule to calculate the probabilities an experiment will yield its various outcomes with.  Using the SIC representation, the description of the measuring apparatus becomes an ordinary set of conditional probabilities, $r(j|i)$.  For instance, for a POVM defined by the set of effects
\begin{equation}
\{E_1,\ldots,E_n\},\ \sum_j E_j = I,
\end{equation}
the Born rule tells us the probabilities $q(j)$ for its outcomes are
\begin{equation}
q(j) = \Tr (\rho E_j),
\end{equation}
but this can be reexpressed as
\begin{equation}
q(j) = \sum_i \left[(d+1) p(i) - \frac{1}{d}\right] r(j|i),
\label{eq:SICMeasProbs}
\end{equation}
where
\begin{equation}
r(j|i) = \Tr (E_j\Pi_i)
\end{equation}
meets the criteria for a conditional probability distribution.

In Ref.\ \cite{fuchs2013}, the simple form in Eq.~(\ref{eq:SICMeasProbs}) was considered so evocative of the usual law of total probability from standard probability theory, and seemingly so basic to Peres's ``unperformed experiments have no outcomes'' considerations, that it was dubbed the \emph{urgleichung}---or German for ``primal equation.''

Similarly, if we have a quantum state $\rho$ encoding our expectations for the SIC measurement on some system at time $t = 0$, we can evolve that state forward to deduce what we should expect at a later time, $t = \tau$.  In textbook language, we relate these two quantum states by a quantum channel---in the simplest case, by a unitary operation:
\begin{equation}
\rho' = U \rho U^\dag.
\end{equation}
Let the SIC representation of $\rho$ be $p(i)$, and let the SIC representation of $\rho'$ be $p'(j)$.  We translate the unitary $U$ into SIC language by calculating
\begin{equation}
u(j|i) = \frac{1}{d}\Tr\left(U\Pi_i U^\dag \Pi_j\right).
\end{equation}
The object $u$ is a $d^2 \times d^2$ doubly stochastic matrix~\cite{DStoch}.  But now, something fascinating happens.  The two quantum states $p(i)$ and $p'(j)$ are related according to
\begin{equation}
p'(j) = \sum_i \left[(d+1) p(i) - \frac{1}{d}\right] u(j|i),
\label{eq:SICTimeEvol}
\end{equation}
an expression identical in form to Eq.\ (\ref{eq:SICMeasProbs}).

Formulas (\ref{eq:SICMeasProbs}) and (\ref{eq:SICTimeEvol}) may be compared with {\it what would have been given\/} by the standard law of total probability
\begin{equation}
q(j) = \sum_i p(i) r(j|i),
\label{eq:ClassMeasProbs}
\end{equation}
and the standard rule for stochastic evolution,
\begin{equation}
p'(j) = \sum_i p(i) u(j|i),
\label{eq:ClassTimeEvol}
\end{equation}
were they applicable. This emphasizes again that the quantum laws are different but, in the setting of a SIC-induced simplex, intriguingly similar to their classical counterparts.

This leads one to wonder whether, or to what extent, these very special forms Eqs.\ (\ref{eq:SICMeasProbs}) and (\ref{eq:SICTimeEvol}) might imply the very arena $\qp$ in which they are valid.  This is the program laid out in Refs.\ \cite{fuchs2013,fuchs2009,Voldemort} and a key motivation for the geometric studies of Refs.\ \cite{appleby2011,appleby2011b,GroupAlg}.  Here we will carry the program much further than previously.

Another familiar operation in the standard language of quantum theory is the Hilbert--Schmidt inner product between two quantum states, $\Tr(\rho\sigma)$.  Using the SIC representations of $\rho$ and $\sigma$ as probability vectors ${p}$ and ${s}$, it is straightforward to show that
\begin{equation}
\Tr(\rho\sigma) = d(d+1) \inprod{p}{s} - 1.
\end{equation}
Because the inner product of any two quantum states $\rho$ and $\sigma$ is bounded between 0 and 1, we know that
\begin{equation}
\frac{1}{d(d+1)} \leq \inprod{p}{s}
 \leq \frac{2}{d(d+1)}.
\label{eq:p-dot-s-bounds}
\end{equation}
We designate these the \emph{fundamental inequalities.}  The upper
bound is simply the quadratic constraint we saw already in
Eq.~(\ref{eq:purity1}), but the lower bound imposes new and
surprisingly intricate conditions on the vectors that can be
admissible states.

We will say that two vectors ${p}$ and ${s}$ in the probability
simplex $\Delta_{d^2}$ are \emph{consistent} if their inner product
obeys both inequalities in Eq.~(\ref{eq:p-dot-s-bounds}).  If we have
a subset of the probability simplex in which every pair of vectors
obeys those bounds, we call it a \emph{germ}: It is an entity from
which a larger structure can grow. If including one additional vector
in a germ could make that set inconsistent, then that germ is said to
be a \emph{maximal.}  We will see that a maximal germ is one way to
define a \emph{qplex.}

Any quantum state space in SIC representation is a qplex.  However, the converse is not true:  There exist qplexes that are not equivalent to quantum state space.  That said, any qplex is already a mathematically rich structure.  A primary goal of this paper is to use that richness and identify an extra condition which can be imposed upon a qplex, such that satisfying that constraint will make the qplex into a quantum state space.

In Section~\ref{sec:basic} we see how quantum physics furnishes a new way that probability assignments can mesh together, a way not foreseen in classical thinking.  This will lead us from very general considerations to the specific definition of a qplex. In Section~\ref{sec:polarity} we apply a tool from the theory of polytopes~\cite{grun,zieg} to derive a number of basic results about the geometry of an arbitrary \qplex.  Among other applications, we find a simple, intuitively appealing proof that a polytope embedded in quantum state space cannot contain the in-sphere of quantum state space.

Sections~\ref{sec:intg} and~\ref{sec:qgroups} are the core of the paper.   In almost every geometrical problem, a study of the symmetries of the object or objects of interest plays an essential role.  However, it turns out that \qplexes\ have the unusual property that the symmetry group, instead of having to be imposed from the outside, is contained internally to the structure.  In this they might be compared with elliptic curves~\cite{elliptic3}.  In spite of the extreme simplicity of the defining equation
\be
y^2 = x^3 + a x + b,
\ee
elliptic curves have managed to remain at the cutting edge of mathematics for two millennia, from the work of Diophantus down to the present day. They play an important role in, for example, the recent proof of Fermat's last theorem~\cite{wiles}. One of the reasons for their high degree of mathematical importance is the fact that they carry within themselves a concealed group.  \Qplexes\ have a similar property.  In Sections~\ref{sec:intg} and~\ref{sec:qgroups} we describe this property, and examine its implications.

In Section~\ref{sec:intg} we present our main application.  We apply the results established in the previous section to the SIC existence problem and show that SIC existence in dimension $d$ is equivalent to the existence of a certain kind of subgroup of the real orthogonal group in dimension $d^2-1$.  We  presented this result in a previous publication~\cite{GroupAlg}, where we derived it by more conventional means.  In this paper, we describe the way we originally proved it, using the \qplex\ formulation.  This is because we believe the method of proof is at least as interesting as the result itself.  

In Section~\ref{sec:character} we turn to the problem of identifying the ``missing assumption'' which will serve to pick out quantum state space uniquely from the set of all \qplexes.  Of course, as is usual in such cases, there is more than one possibility.  We identify one such assumption:  the requirement that the symmetry group contain a subgroup isomorphic to the projective unitary group.  This is a useful result because it means that we have a complete characterization of quantum state space in probabilistic terms.  It also has an important corollary:  That SIC existence in dimension $d$ is equivalent to the existence of a certain kind of subgroup of the real orthogonal group in dimension $d^2-1$.

Finally, we wrap up in Section~\ref{sec:future} with list of several possible directions for future investigations.  If this research program is on the right track, it is imperative that a more basic path from qplex to quantum state space be found.  There is plenty of work to do here.

\section{The Basic Scheme}
\label{sec:basic}
The urgleichung (\ref{eq:SICMeasProbs}) and the inequalities
(\ref{eq:p-dot-s-bounds}) are not independent.  In this section, we
will start with a generalized form of the urgleichung and, making a
few additional assumptions, derive the fundamental inequalities.  This
is, strictly speaking, not necessary for the mathematical developments
in the later sections of the paper.  One can assume the fundamental
inequalities as a starting point and then proceed from that premise.
In fact, we will later see that using that approach, one can derive as
consequences the assumptions we will invoke here.  Speaking in general
terms, we can think of this section as proving the ``if'' direction,
and the following section as proving ``only if.''  One benefit of
deriving the fundamental inequalities in this manner is to help
compare and contrast our reconstruction of quantum theory with other
approaches~\cite{transcript, coecke2016, HoehnWever, Masanes, Hardy01,
  Schack03, Barnum}.  These other reconstructions are
\emph{operational} in character: They take, as fundamental conceptual
ingredients, laboratory procedures like ``preparations'' and
``tests.''  Our language in this section will have a similar tone.
However, we will keep Peres' dictum that ``unperformed experiments
have no results'' at the forefront of our considerations.

Our first step is to understand how the urgleichung is an example of
this principle.  To do so, we consider the following
scenario~\cite{Voldemort, transcript}.

\begin{figure}[ht]
\includegraphics[width=8cm]{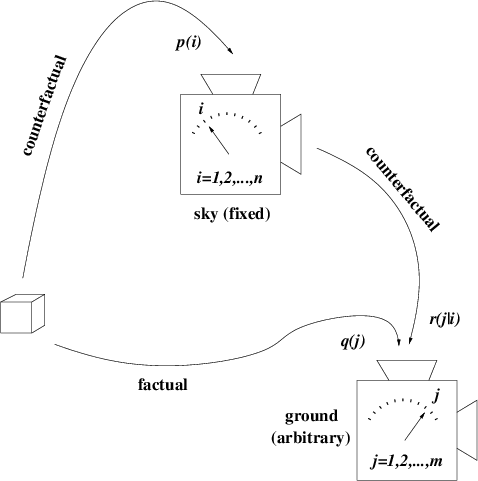}
\caption{\label{fig:ground-versus-sky} Analysing one scenario in terms
  of another: An agent Alice intends to perform an experiment on the
  ground, whose outcomes she labels with the index $j$.  The other
  index, $i$, labels the outcomes of a ``Bureau of Standards''
  measurement which Alice \emph{could} carry out, but which remains
  unperformed.  Classical physics and quantum physics both allow for
  Bureau of Standards measurements, experiments that are
  \emph{informationally complete} in the following sense.  If Alice
  has a set of probabilities $p(i)$ for the Bureau of Standards
  measurement outcomes, she can calculate the proper set of
  probabilities $q(j)$ for the outcomes of the ground measurement,
  using the conditional probabilities $r(j|i)$.}
\end{figure}

Fix a dimension $d \geq 2$, and consider a system to which we will
ascribe a quantum state in $d$-dimensional Hilbert space
$\mathcal{H}_d$.  We will investigate this system by means of two
measuring devices, which we model in the standard way by POVMs.  One
measuring device is a SIC measurement, defined by a set of $d^2$
rank-1 projection operators $\{\Pi_i\}$.  The effects which comprise
this POVM are the operators rescaled by the dimension:
\begin{equation}
E_i = \frac{1}{d} \Pi_i.
\end{equation}
We will refer to this as the ``Bureau of Standards'' measurement.  It
is helpful to imagine this measuring device as being located in some
comparatively inaccessible place: perhaps inside a vault, or secured
in an airship floating through the sky.  An agent \emph{can} take her
system of interest to the Bureau of Standards device, but she has good
reason to want to bypass that step.  The other measurement is an
arbitrary POVM, whose effects we denote by~$F_j$.

As illustrated in Figure~\ref{fig:ground-versus-sky}, we will consider
two experimental scenarios, which we will call the ``ground path'' and
the ``sky path.''  If we follow the ground path, we take our system of
interest directly to the $\{F_j\}$ measuring device, which we will
call the measurement on the ground.  If we instead follow the sky
path, we will take our system to the Bureau of Standards measurement,
physically obtain a result by performing that measurement, and then
come back down for the second stage, where we conduct the measurement
on the ground.

Suppose that Alice follows the sky path in
Figure~\ref{fig:ground-versus-sky}.  That is, she physically takes her
system of interest and performs the Bureau of Standards measurement
upon it.  Then, she returns the system to the ground and conducts the
measurement $\{F_j\}$ there.  Before carrying out the Bureau of
Standards measurement, she has some expectations for what might
happen, which she encodes as a probability distribution $p(i)$.
Obtaining an outcome $i$, she updates her state assignment for the
system to the operator $\Pi_i$.  Her expectations for the outcome of
the ground measurement will then be the conditional probabilities
$r(j|i)$.  Prior to performing the Bureau of Standards measurement,
Alice assigns the probability
\begin{equation}
\hbox{Prob}(j) = \sum_i p(i) r(j|i)
\label{eq:naive-LoTP}
\end{equation}
to the event of obtaining outcome $j$ when she brings the system back
down to the ground and performs the second measurement in the sequence.

Classical intuition suggests that Alice should use the same expression
for computing the probability of outcome $j$ on the ground even if she
goes directly to the ground experiment and does not perform the
measurement in the sky.  If $p(i)$ is the probability that she
\emph{would} obtain outcome $i$ \emph{were she to perform} the sky
measurement, and $r(j|i)$ is the conditional probability for outcome
$j$ \emph{if} the event $i$ \emph{were to occur} in the sky, then it
is almost instinctive to calculate the probability of~$j$ by summing
$p(i) r(j|i)$.  Mathematically, this is \emph{not necessarily
  correct,} because the ground path and the sky path are two different
physical scenarios.  If $C_1$ and $C_2$ are two background conditions,
then nothing in probability theory forces $\hbox{Prob}(j|C_1) =
\hbox{Prob}(j|C_2)$.  Writing $q(j)$ for the probability of obtaining
$j$ by following the ground path, we have that
\begin{equation}
q(j) \hbox{ is not necessarily equal to }
 \sum_i p(i) r(j|i).
\end{equation}
It is merely the assumption that an informationally complete
measurement must be measuring some pre-existing physical property of
the system that leads Alice to use Eq.~(\ref{eq:naive-LoTP}) even when
she does not physically obtain an outcome in the sky.  In other words,
using Eq.~(\ref{eq:naive-LoTP}) to calculate $q(j)$ amounts to
assuming that the measurement outcome $i$ is \emph{as good as
  existing,} even when it remains completely counterfactual.

Probability theory itself does not tell us how to find $q(j)$ in terms
of~$p(i)$ and $r(j|i)$.  Classical intuition suggests one way of
augmenting the abstract formalism of probability theory:  using
Eq.~(\ref{eq:naive-LoTP}).  The crucial point is that \emph{quantum
  theory gives us an alternative.}  It is simply to use the Born rule,
in the form of the urgleichung.

The Born-rule probability for obtaining the outcome with
index $j$ is
\begin{equation}
q(j) = \Tr(\rho F_j) = \sum_i \left[(d+1) p(i) - \frac{1}{d}\right]
r(j|i),
\end{equation}
where
\begin{equation}
r(j|i) = \Tr (\Pi_i F_j).
\end{equation}
Note that $r(j|i)$ is also the probability that the Born rule would
tell us to assign to the outcome $j$ if our quantum state assignment
for the system were $\Pi_i$.

Probability theory is a way to augment our raw experiences of life:
It provides a means to manage our expectations carefully.  In turn,
quantum theory augments the mathematics of probability, furnishing
links between quantities that, considering only the formalism of
probability theory, would be unrelated.  These new relationships are
quantitatively precise, but at variance with classical intuition,
reflecting the principle that unperformed experiments have no outcomes.

We now explore the consequences of relating mutually exclusive
hypothetical scenarios by the urgleichung.  Using seven assumptions, of
which the urgleichung is the most radical, we will arrive at the
fundamental inequalities (\ref{eq:p-dot-s-bounds}).  Because the
constants $d^2$ and $d+1$ and $1/d$ look rather arbitrary at first
glance, we will begin with a more general expression.

\begin{assump}
\label{assump:urgleichung}
The Generalized Urgleichung.  Given a Bureau of Standards probability
distribution $\{p(i): i = 1,\ldots,N\}$, and a matrix of conditional
probabilities $r(j|i)$, we compute the probabilities for an experiment
on the ground by means of
\begin{equation}
q(j) = \sum_{i=1}^N \left[\alpha p(i) - \beta\right] r(j|i).
\label{eq:gen-urgleichung}
\end{equation}
\end{assump}
In what follows, this will be our primary means of relating one
probability distribution to another.  The basic normalization
requirements are
\begin{equation}
\sum_i p(i) = 1,\ \sum_j r(j|i) = 1,\ \sum_j q(j) = 1.
\end{equation}
Normalization relates the constants $\alpha$, $\beta$ and $N$:
\begin{equation}
\alpha = N\beta + 1.
\end{equation}

We denote the set of valid states ${p}$ by~$\cP$, and the set of
valid measurements by~$\cR$.  For any ${p} \in \cP$ and any
$r(j|i) \in \cR$, the vector ${q}$ calculated using the
urgleichung is a proper probability distribution.

If we take any $r \in \cR$ and sum over both indices, we find that
\begin{equation}
\sum_{i,j} r(j|i) = \sum_i 1 = N.
\end{equation}

\begin{assump}
\label{assump:max}
Maximality.  The set of all states $\cP$ and the set of all measurements $\cR$ together have the property that no element can be added to either without introducing an inconsistency, i.e., a pair $(p \in \cP, r \in \cR)$ for which the urgleichung yields an invalid probability.
\end{assump}

It is sometimes helpful to write the urgleichung in vector notation:
\begin{equation}
{q} = rM{p}.
\end{equation}
Here, $r$ is a matrix whose $(j,i)$ entry is given by $r(j|i)$, and
$M$ is a linear combination of the identity matrix $I$ and the matrix
whose elements all equal 1, the so-called Hadamard identity $J$:
\begin{equation}
M = \alpha I - \beta J.
\end{equation}

Assumptions \ref{assump:urgleichung} and \ref{assump:max} imply a fair bit about the structure of~$\cP$ and $\cR$.
\begin{lemma}
\label{lm:convex}
The set $\cP$ of all states and the set $\cR$ of all measurements are both convex and closed.
\end{lemma}
\begin{proof}
Let $p_1,p_2 \in \cP$, and for any $r \in \cR$, define
\ea{
q_1 &= rMp_1, & q_2 &= rMp_2.
}
By assumption, both $q_1$ and $q_2$ are valid probability vectors (i.e., they are normalized, and all their entries are nonnegative).  Define
\begin{equation}
p_\lambda = \lambda p_1 + (1-\lambda)p_2.
\end{equation}
Then
\begin{equation}
q_\lambda = rMp_\lambda = \lambda q_1 + (1-\lambda)q_2.
\end{equation}
This is a convex combination of points in the probability simplex, and as such it also belongs to the probability simplex.  By assumption, this holds true for every $r \in \cR$, and so by maximality, $p_\lambda \in \cP$.  The proofs of the convexity of~$\cR$ and of closure work analogously.
\end{proof}

Consider the case where the ground and sky measurements are the same.
In that scenario, we have ${q} = {p}$, and so the measurement
matrix must be the inverse of~$M$:
\begin{equation}
r_F = M^{-1} = \frac{1}{\alpha}I + \frac{\beta}{\alpha}J.
\label{eq:r_F}
\end{equation}
Note that we have to include $r_F$ within $\cR$ by the maximality assumption.

The urgleichung is one way that quantum theory builds upon the
mathematics of probability, interconnecting our previsions for
different experiments, previsions that basic probability theory alone
would leave separate.  Quantum theory augments the probability
formalism in another fashion as well, and it is to that which we now
turn.

Our next assumption will establish that the set of measurements $\cR$
can be constructed from the set of states $\cP$.  On a purely
mathematical level, we could justify this by saying that we wish to
build the most parsimonious theory possible upon the urgleichung, and
so we simplify matters by having one fundamental set instead of two.
As far as constructing a mathematical theory goes, this is certainly a
legitimate way to begin.  We can, however, provide a more physical
motivation than that.

Probability theory, intrinsically, assumes very little about the
structure of event spaces.  With it, we can for example discuss
rolling a die and recording the side that lands facing up; we say that
the realm of possible outcomes for this experiment is the set
$\{1,2,3,4,5,6\}$.  In this experiment, the outcome ``1'' is no more
\emph{like} the outcome ``2'' than it is \emph{like} the outcome
``6''.  We can ascribe probabilities to these six potential events
without imposing a similarity metric upon the realm of outcomes.  We
use integers as labels, but we care hardly at all about the
number-theoretic properties of those integers.  When we roll the die,
we are indifferent to the fact that 5 is prime and 6 is perfect.  Nor
is the event of observing a particular integer in this experiment
related, necessarily, to the event of observing that same integer in a
\emph{different} experiment.

When Alice first learns probability theory, she picks up this habit of
tagging events with integers.  If Alice considers a long catalogue of
experiments that she could perform, she might label the possible
outcomes of the first experiment by the integers from 1 to~$N_1$, the
outcomes of the second experiment by the integers $\{1,\ldots,N_2\}$
and so on.  But, in general, Alice has the freedom to permute these
labels as she pleases.  She does not have to regard the experience of
obtaining $j = 17$ in one experiment as similar to the experience of
obtaining $j = 17$ in any other.

But what if Alice wants more structure than this?  When Alice
contemplates an experiment that she might carry out, she considers a
set of possible outcomes for it, \emph{i.e.,} a realm of potential
experiences which that action might elicit.  She can assign each of
those potential experiences a label drawn from whatever mathematical
population she desires.  Her \emph{index set} for a given experiment
can be a subset of whatever population she finds convenient.  When
Alice adopts the urgleichung as an empirically-motivated addition to
the bare fundamentals of probability theory, does she, by that act,
also gain a natural collection of mathematical entities from which to
build index sets?

In fact, she has just such a collection at hand: She can use the set
of valid states, $\cP$!

To consider the matter more deeply, we ask the following question:
Under what conditions would Alice consider two outcomes of two
different experiments to be equivalent?  For example, Alice
contemplates two experiments she might feasibly perform, which she
describes by two matrices $r$ and $r'$.  When would Alice treat an
outcome $j$ of experiment $r$ to be equivalent to an outcome $j'$
of~$r'$~\cite{fuchs2002}? Generally, the tools she has on hand to make
such a judgment are her probability ascriptions for those outcomes.
If her overall mesh of beliefs is that her probability of experiencing
$j$ upon enacting $r$ is the same as her probability for finding $j'$
when enacting $r'$, no matter what her state assignment ${p}$,
then she has good grounds to call $j$ and $j'$ equivalent.  In order
to satisfy $q(j) = q'(j')$ for all $p \in \cP$, the measurement
matrices $r$ and $r'$ must obey
\begin{equation}
r(j|i) = r'(j'|i),\ \forall i.
\end{equation}

The simplest way to ensure that this is possible is to build all
elements $r$ of the set $\cR$ from a common vocabulary.  When we
construct an element $r \in \cR$, we draw each row from a shared pool
of ingredients.  The natural, parsimonious choice we have on hand for
this purpose is the set $\cP$.  This means that, up to scaling,
measurement outcomes are actually identified with points in the
probability simplex.

Let $r \in \cR$ be a valid measurement.  If each row of the matrix
$\{r(j|i)\}$ can also naturally be identified with a vector ${s} \in
\cP$, then we are led to consider the vector ${s}$ sitting inside $r$
in some fashion.  The simplest reasonable relation between ${s}$,
which is a vector with $N$ elements, and the measurement matrix $r$,
whose rows have length $N$, is to have a row of~$r$ be linearly
proportional to~${s}$.

\begin{assump}
\label{assump:R-from-P}
Measurement Matrices are Constructed from States.  Given any $r \in \cR$, we
can write a row $\{r(j|i) : i = 1,\ldots,N\}$ as a vector
${s}_j \in \cP$, up to a normalization factor:
\begin{equation}
r(j|i) = N \gamma_j s_j(i).
\label{eq:R-from-P}
\end{equation}
Furthermore, any state in $\cP$ can be used in this manner.
\end{assump}
For brevity, we will refer to the $s_j$ as ``measurement vectors.''
We will shortly identify the meaning of the constants $\{\gamma_j\}$,
which we have written with the prefactor $N$ for later convenience.

\begin{assump}
\label{assump:ignorance}
Possibility of Maximal Ignorance. The state $c$, defined by
\begin{equation}
c(i) = \frac{1}{N}\ \forall\ i,
\end{equation}
belongs to $\cP$.
\end{assump}
This can be deduced from other postulates, but the state $c$ is a
useful tool, and it is helpful to point its existence out explicitly.
For example, substituting the state of complete ignorance ${c}$ into
the urgleichung, we obtain
\begin{equation}
q(j) = \frac{1}{N} \sum_i r(j|i).
\end{equation}

What is the meaning of the factors $\{\gamma_j\}$?  To find out, we apply
a measurement $r \in \cR$ to the state ${c}$:
\begin{equation}
q(j) = \frac{1}{N} \sum_i r(j|i)
 = \gamma_j \sum_i s_j(i)
 = \gamma_j.
\end{equation}
The factors $\{\gamma_j\}$ indicate the probability of obtaining the
$j^{\rm th}$ outcome on the ground when the agent is completely
indifferent to the potential outcomes of the sky experiment.

If the effect of some $r \in \cR$, when applied via the urgleichung,
is to send ${c}$ to itself, then we have that
\begin{equation}
c(j) = \frac{1}{N} = \frac{1}{N} \sum_i r(j|i)
\Rightarrow \sum_i r(j|i) = 1.
\end{equation}
Combined with the basic normalization requirement for conditional
probabilities, this states that a measurement that preserves
${c}$ is represented by a \emph{doubly stochastic} matrix.

\begin{lemma}
\label{lm:doubly-stochastic}
Measurements that send the state $c$ to itself are represented by doubly stochastic matrices.
\end{lemma}

When we postulated the urgleichung, we added structure to the bare
essentials of probability theory, and the structure we added related
one experiment to another in a way above and beyond basic coherence.
With Assumption~\ref{assump:R-from-P}, we are also interrelating
different experiments.  We can appreciate this in another way by
considering what it means for a physical system to be usable as a
scientific instrument.

What conditions must an object meet in order to qualify as a piece of
laboratory apparatus?  Classically, a bare minimum requirement is that
the object has a set of distinguishable configurations in which it can
exist.  These might be positions of a pointer needle, heights of a
mercury column, patterns of glowing lights and so forth.  The
essential point is that the system can be in different configurations
at different times: A thermometer that always reports the same
temperature is useless.  We can label these distinguishable
configurations by an index $j$.  The \emph{calibration} process for a
laboratory instrument is a procedure by which a scientist assigns
conditional probabilities $r(j|i)$ to the instrument, relating the
readout states $j$ to the inputs $i$.  In order to make progress, we
habitually assume that nature is not so perverse that the results of
the calibration phase become completely irrelevant when we proceed to
the next step and apply the instrument to new systems of unknown
character.

But what if nature \emph{is} perverse?  Not enough so to forbid the
possibility of science, but enough to make life interesting.
Quantitatively speaking, what if we must modify the everyday
assumption that one can carry the results of a calibration process
unchanged from one experimental context to another?

\emph{The urgleichung is just such a modification.}  The $\{r(j|i)\}$
do not become irrelevant when we move from the sky context to the
ground, but we do have to use them in a different way.

In quantum physics, we no longer treat ``measurement'' as a passive
reading-off of a specified, pre-existing physical quantity.  However,
we do still have a counterpart for our classical notion of a system
that can qualify as a laboratory apparatus.  Instead of asking whether
the system can exist in one of multiple possible classical states, we
ask whether our overall mesh of beliefs allows us to consistently
assign any one of multiple possible catalogues of expectations.  That
is, if an agent Alice wishes to use a system as a laboratory
apparatus, she must be able to say now that she can conceive of
ascribing any one of several states to it at a later time.  We define
a \emph{discrete apparatus} as a physical system with an associated
set of states,
\begin{equation}
\{s_1,\ldots,s_{m}\} \subset \cP.
\end{equation}
The analogue of classical uncertainty about where a pointer might be
pointing is the convex combination of the states $\{s_j\}$.
Therefore, our basic mental model of a laboratory apparatus is a
polytope in~$\cP$, with the $\{s_j\}$ as its vertices.  Assumption
\ref{assump:R-from-P} says that \emph{Alice can pick up any such
  apparatus and use it as a ``prosthetic hand'' to enrich her
  experience of asking questions of nature.}

We can think of Assumption~\ref{assump:R-from-P} in another way, if we rewrite Eq.~(\ref{eq:R-from-P}) in the following manner:
\begin{equation}
s_j(i) = \frac{\left(\frac{1}{N}\right) r(j|i)}{\gamma_j}.
\end{equation}
Earlier, we noted that $\gamma_j$ is the probability of obtaining the $j^{\rm th}$ outcome on the ground, given complete ignorance about the potential outcomes of the sky experiment.  In addition, $1/N$ is the probability assigned to each outcome of the sky experiment by the state of complete ignorance.  So,
\begin{equation}
s_j(i) = \frac{\hbox{PrCI}(i)\, r(j|i)}{\hbox{PrCI}(j)},
\end{equation}
where the notation ``PrCI'' here indicates a probability assignment given that the state for the sky experiment is $c$.  Note that $\hbox{PrCI}(j|i) = r(j|i)$.  But this means that the expression on the right-hand side above is just the ordinary Bayes formula for inverting conditional probabilities:
\begin{equation}
\hbox{PrCI}(i|j) = \frac{\hbox{PrCI}(i)\, \hbox{PrCI}(j|i)}{\hbox{PrCI}(j)}.
\end{equation}
Therefore, we can interpret the mathematical relation established in Assumption~\ref{assump:R-from-P} as saying that ``posteriors from maximal ignorance are priors''~\cite{fuchs2013}.  For the remainder of this paper, we will not be considering in detail the rules for changing one's probabilities upon new experiences---a rather intricate subject, all things told~\cite{QBist-decoherence, stacey-thesis}.  So, we will not stress the ideas of ``priors'' and ``posteriors,'' but it is good to know that this reading of Assumption~\ref{assump:R-from-P} exists.

Writing the urgleichung in terms of the vector ${s}_j$,
\begin{align}
q(j) &= \sum_i \left[\alpha p(i) - \beta\right] N \gamma_j
        s_j(i) \\
 &= N\alpha \gamma_j \inprod{p}{s_j} - N\beta\gamma_j.
\end{align}
The fact that $q(j)$ must be nonnegative for all $j$ implies a lower
bound on the scalar product $\inprod{p}{s_j}$:
\begin{equation}
\inprod{p}{s_j} \geq \frac{\beta}{\alpha}.
\label{eq:first-lower-bound}
\end{equation}

The measurement described by the matrix $r_F$ in Eq.~(\ref{eq:r_F})
yields, by construction, equal probabilities for all outcomes given
the input state ${c}$.  That is, it is an experiment with $N$
outcomes, and $\gamma_j = 1/N$ for all of them.  Therefore, we can
take the rows of~$r_F$ as specifying $N$ special vectors within~$\cP$.
We have that
\begin{equation}
r_F(j|i) = e_j(i),
\end{equation}
where the vector ${e}_j$ is flat across all but one entries:
\begin{equation}
e_j(i) = \frac{1}{\alpha}(\delta_{ji} + \beta).
\end{equation}
We will refer to the vectors $\{{e}_k\}$ as the \emph{basis
  distributions.}

What happens if we take a measurement $r \in \cR$, and act with it via
the urgleichung upon a basis distribution ${e}_k$?  The result is
straightforwardly computed to be
\begin{align}
q(j) &= \sum_i \left[\alpha\left(\frac{\beta}{\alpha}
                                + \frac{1}{\alpha} \delta_{ik}
                          \right)
                     - \beta
               \right] r(j|i) \\
 &= \beta \sum_i r(j|i) + \sum_i \delta_{ik} r(j|i)
    - \beta \sum_i r(j|i) \\
 &= r(j|k).
\label{eq:r-upon-basis}
\end{align}
This will be useful later.

Note that the basis distributions all have magnitude equal to
\begin{equation}
\inprod{e_k}{e_k} = \frac{1 + 2\beta + N\beta^2}{\alpha^2}.
\label{eq:basis-purity}
\end{equation}
This result singles out a \emph{distinguished length scale} in
probability space, namely, the radius of the sphere on which all the
basis distributions live.

The lower bound (\ref{eq:first-lower-bound}) suggests the following construction.  Let $H$ be the hyperplane of vectors in~$\mathbb{R}^N$ that sum to unity:
\be
H =\left\{v \in \fd{R}^{N} \colon  \inprod{v}{c} = \frac{1}{N}\right\}.
\ee
This hyperplane includes the probability simplex.  For any set $A$ of probability distributions, consider the set
\be
\dl{A} = \left\{ u \in H \colon  \inprod{u}{v} \ge \frac{\beta}{\alpha} \ \forall v \in A\right\}.
\ee
This set includes all the probability distributions that are consistent with each point in~$A$, with respect to the lower bound we derived from the urgleichung.  We will designate the set $\dl{A}$ the \emph{polar} of~$A$, following the terminology for a related concept in geometry~\cite{grun, zieg}.  Let $\cP$ be the set of all valid states.  The set of all measurement vectors that are consistent with these states, with respect to the lower bound, is that portion of the polar of~$\cP$ that lies within the probability simplex:
\be
\dl{\cP} \cap \Delta = \left\{ s : \inprod{s}{p} \geq \frac{\beta}{\alpha} \forall p \in \cP \right\} \cap \Delta.
\ee
If some $s$ in this set is not in the set $\cP$, then some measurement vector does not correspond to a state.  Likewise, if some $p \in \cP$ is not in this set, then that state cannot correspond to a measurement vector.  Both of these cases violate the mapping we have advocated on general conceptual grounds.  Therefore, our first three assumptions imply that we consider sets $\cP$ for which
\be
\cP = \dl{\cP} \cap \Delta.
\ee
We will see momentarily how to simplify this condition, establishing the condition that a state space $\cP$ must be self-polar:
\begin{equation}
\cP = \dl{\cP}.
\end{equation}

In order to prove this proposition, we need to know more about the operation of taking the polar.  We can derive the relations we require by adapting some results from the higher-dimensional geometry literature.  Gr\"unbaum~\cite{grun} defines the polar of $A\subseteq
\fd{R}^{d^2}$ to be the set
\be
A^{\circ} = \{ u \in \fd{R}^{d^2} \colon \inprod{u}{v} \le 1 \ \forall v \in A\}.
\ee
Our definition of the polar $\dl{A}$ is close enough to this definition of~$A^{\circ}$ that many results about the latter can be carried over with little effort.
The properties of the polar $\dl{A}$ are summarized in the following theorem.
\begin{theorem}
\label{tm:polarity}
For all $A\subseteq H$, the polar $\dl{A}$ is a closed, convex  set containing $c$.  Since we will frequently be invoking the concept of convex hulls, we introduce the notation $\cc(A)$ for the closed, convex hull of the set $A$.  We have
\ea{
\dl{A} & =  \dl{\bigl(\cc(A\cup \{c\})\bigr)},
\\
\ddl{A} &= \cc(A\cup \{c\}),
}
for all $A\subseteq H$.  In particular, $A$ is equal to its double polar $\ddl{A}$ if and only if it is closed, convex and contains $c$.

For  all $A$, $B\subseteq H$
\be
A\subseteq B \implies \dl{B} \subseteq \dl{A}.
\ee
If $\mathcal{A}$ is an arbitrary family of subsets of $H$ then
\ea{
\dl{\left(\bigcup_{A\in \mathcal{A}} A \right)} &= \bigcap_{A\in \mathcal{A}} \dl{A}.
\label{eq:dualUnionB}
\\
\intertext{If, in addition, $\ddl{A}=A$ for all $A\in\mathcal{A}$ then}
\dl{\left(\bigcap_{A\in \mathcal{A}} A \right)} &=\cc\left( \bigcup_{A\in \mathcal{A}} \dl{A}\right).
\label{eq:dualIntersectionB}
}
\end{theorem}
\begin{proof}
All these properties follow by relating Gr\"unbaum's definition of the
polar with ours.  Let  $f\colon \fd{R}^{N} \to \fd{R}^{N}$ be the affine map defined by
\be
f(u) = N\alpha (u-c),
\ee
and let $H_0$ be the subspace
\be
H_0 = \{ u\in \fd{R}^{N} \colon \la u, c \ra = 0\}.
\ee
One then has
\be
\dl{A} = f^{-1} \Bigl( \bigl(-f(A)\bigr)^{\circ}\cap H_0\Bigr)
\ee
for all $A\subseteq H$.  With this in hand the theorem becomes a
straightforward consequence of textbook results.
\end{proof}

Now, consider the relation $\cP = \dl{\cP} \cap \Delta$, and take the polar of both sides:
\begin{equation}
\dl{\cP} = \dl{\left(\dl{\cP} \cap \Delta\right)} = \cc \left(\ddl{\cP} \cup \dl{\Delta_N}\right).
\end{equation}
We know that $\cP$ is closed and convex, and that it contains the center point $c$.  Therefore,
\begin{equation}
\ddl{\cP} = \cP.
\end{equation}
What is the polar of the probability simplex $\Delta$?  In fact, it is the basis simplex $\bs$.
\begin{lemma}
\label{lm:polar-basis}
The probability simplex and the basis simplex are mutually polar:
\ea{
\dl{\Delta} &= \bs, & \dl{\bs} &= \Delta.
}
\end{lemma}
\begin{proof}
The probability simplex contains normalized vectors, so it lies in the hyperplane $H$, and all of its vectors have wholly nonnegative entries.  Let $v_i$ be the $i^{\rm{th}}$ vertex of $\ps$ (so $v_i(j) = \delta_{ij}$).  Then the probability simplex is
\ea{
\ps = \{ u \in H \colon \la u, v_i\ra \ge 0 \ \forall i\}.
}
Let $f\colon H \to H$ be the affine map defined by
\be
f(u) = \frac{1}{\alpha} u  + \frac{\beta}{\alpha}.
\ee
Then $\bs = f(\ps)$.  It follows that
\ea{
\bs = \left\{u \in H \colon \la u, v_i \ra \ge \frac{\beta}{\alpha}
 \ \forall i \right\}.
}
Taking account of Theorem~\ref{tm:polarity}  we deduce
\ea{
\bs & = \dl{\{v_i \colon i = 1, \dots, N\}}
= \dl{\ps}.
}
The fact that $\dl{\bs} = \ps$ is an immediate consequence of this and the fact that the double polar of a closed convex set is itself (see Theorem~\ref{tm:polarity}).
\end{proof}

\begin{theorem}
A state space $\cP$ satisfying Assumptions \ref{assump:urgleichung}, \ref{assump:max}, \ref{assump:R-from-P} and \ref{assump:ignorance} is self-polar:
\begin{equation}
\cP = \dl{\cP}.
\end{equation}
\end{theorem}
\begin{proof}
We already know that
\begin{equation}
\dl{\cP} = \cc\left(\ddl{\cP} \cup \dl{\Delta}\right),
\end{equation}
and now we can say that
\begin{equation}
\dl{\cP} = \cc(\cP \cup \bs).
\end{equation}
But we established already that $\cP$ always contains the basis distributions, and that $\cP$ is closed and convex.  Therefore, $\cP$ is self-polar.
\end{proof}

The fact that a state space is self-polar implies the existence of two more distinguished length scales.  To see why, it is helpful to work in barycentric coordinates, shifting all our vectors so that the origin lies at the barycenter point of the simplex, the point $c$:
\be
p \to p' = p - c.
\ee
In these coordinates, our lower bound (\ref{eq:first-lower-bound}) becomes
\be
\inprod{p'}{s'} \geq - \frac{1}{N\alpha}.
\ee
Any basis distribution $e_j$ satisfies
\be
\inprod{e_j'}{e_j'} = \frac{N-1}{N\alpha^2}.
\ee
We define the \emph{out-sphere} $\so$ to be the sphere centered on the barycenter with radius
\be
\ro^2 = \frac{N-1}{N\alpha^2}.
\ee
The ball bounded by $\so$ is the \emph{out-ball} $\bo$.  We will see shortly that the polar of the out-ball is a ball centered at the barycenter and having radius
\be
\ri^2 = \frac{1}{N(N-1)}.
\ee
We designate this ball the \emph{in-ball} $\bi$, and its surface is the \emph{in-sphere} $\si$.  Finally, note that if we take
\be
\rmid^2 = \frac{1}{N\alpha},
\ee
any two points both lying within $\rmid$ of the barycenter will be consistent with respect to the bound (\ref{eq:first-lower-bound}).  This defines the \emph{mid-ball} $\midball$ and its surface, the \emph{mid-sphere} $\midsphere$.  It follows that
\be
\ri \ro = \rmid^2.
\ee

We now prove the fact we stated a moment ago.
\begin{lemma}
\label{lm:simpballpolars}
The out- and in-balls are mutually polar:
\ea{
\dl{\bo} &= \bi, & \dl{\bi} &= \bo.
}
\end{lemma}
\begin{proof}
Let $f\colon H\to H$ be the affine map defined by
\be
f(u) = c + \frac{\ro}{\ri} (u-c).
\ee
Then $f(\bi) = \bo$.  Consequently, given arbitrary $u\in H$,
\ea{
&   &u &\in \dl{\bo} &&
\\
&\iff &  \la u, f(v) \ra &\ge \frac{\beta}{\alpha} & \forall v& \in \bi
\\
&\iff & \la u-c, f(v)-c\ra &\ge -\ro \ri &\forall v& \in \bi
\\
&\iff & \la u-c , v-c\ra & \ge -\ri^2 &\forall v&\in \bi
\\
&\iff & u & \in \bi &&
}
So $\dl{\bo}=\bi$.  The fact that $\dl{\bi} = \bo$ is an immediate consequence of this and the fact that the double polar of a closed convex set is itself.
\end{proof}

These distinguished length scales suggest another assumption we ought to make about our state space.  Earlier, we stated that the barycenter $c$ must belong to our set of admissible probability distributions.  It is natural to ask how far away from complete ignorance we can go before we encounter complications.  Can our state space $\cP$ contain all the points in a little ball around~$c$?  Intuitively, it is hard to see why not.  How big can we make that ball around the center point $c$ before we run into trouble?  The simplest assumption, in this context, is to postulate that the first complication we encounter is the edge of the probability simplex itself.  Where does a sphere centered at~$c$ touch the faces of the simplex?  The center of a face of the probability simplex is found by
taking the average of $N-1$ of its vertices:
\begin{equation}
\bar{v}_k(i) = \frac{1}{N-1}(1 - \delta_{ik}).
\end{equation}
The sphere centered on $c$ that just touches these points has a radius
given by
\begin{equation}
(\bar{v}_k - c)^2 = \frac{1}{N(N-1)}.
\end{equation}
The in-sphere $\si$ is just the \emph{inscribed} sphere of the probability simplex.

\begin{assump}
\label{assump:upper-bound}
Every state space $\cP$ contains the in-ball.
\end{assump}

Because the polar of the in-ball is the out-ball, and polarity
reverses inclusion, it follows that every self-consistent state space
is bounded by the out-sphere.  This result has the form of an
``uncertainty principle'': It means that our probability distributions
can never become too narrowly focused.  For any two points $p$ and $s$ within our state space $\cP$, we have
\be
L \leq \inprod{p}{s} \leq U,
\ee
where the lower and upper bounds are given by
\begin{align}
L &= -\frac{1}{N\alpha} + \frac{1}{N}, \\
U &= \frac{N-1}{N\alpha^2} + \frac{1}{N}.
\end{align}

Recall from
Lemma~\ref{lm:polar-basis} that the polar of the probability simplex is the simplex defined by the basis distributions $e_k$,
which in barycentric coordinates is seen to be the probability simplex
rescaled:
\be
e_k'(i) = e_k(i) - c(i) = \frac{1}{\alpha}\left(\delta_{ik} - c(i)\right).
\ee

Call two extremal states $p$ and $s$ in a state space \emph{maximally distant} if
they saturate the lower bound:
\be
\inprod{p'}{s'} = -\frac{1}{N\alpha}.
\ee
Let
\begin{equation}
\{p'_k : k = 1,\ldots,m \}
\end{equation}
be a set of Mutually Maximally Distant (MMD) states.  That is, for all $k$,
\begin{equation}
\inprod{p_k'}{p_k'} = \ro^2,
\end{equation}
and for $k \neq l$,
\begin{equation}
\inprod{p_k'}{p_l'} = -\rmid^2.
\end{equation}
Construct the vector quantity
\begin{equation}
V = \sum_k p_k'.
\end{equation}
From the fact that the magnitude $\inprod{V}{V} \geq 0$, it follows that
\begin{equation}
m \leq 1 + \frac{\ro^2}{\rmid^2}.
\end{equation}
Substituting in the definitions of the radii, we arrive at the
relation
\begin{equation}
m \leq 1 + \frac{N-1}{\alpha}.
\end{equation}
Let us now make an assumption:  We want this bound to be attainable.

\begin{assump}
\label{assump:m-max}
A state space $\cP$ contains an MMD set of size
\be
m_{\rm max} = 1 + \frac{N-1}{\alpha}.
\label{eq:m-max}
\end{equation}
\end{assump}
Note that both $N$ and $m_{\rm max}$ are positive integers by
assumption.  This means that $\alpha$ must divide $N-1$ neatly. 

To set the context for our next assumption, switch back to the
original frame.  Recall that any two points $p$ and $s$ within our
state space $\cP$ satisfy
\be
L \leq \inprod{p}{s} \leq U,
\ee
where the lower and upper bounds are given by
\begin{align}
L &= -\frac{1}{N\alpha} + \frac{1}{N}, \\
U &= \frac{N-1}{N\alpha^2} + \frac{1}{N}.
\end{align}
Comparing these two quantities, and using Eq.~(\ref{eq:m-max}) to
simplify, we obtain
\begin{equation}
\frac{U}{L} = 1 + \frac{m_{\rm max}}{\alpha - 1},
\end{equation}
where $m_{\rm max}$ is a positive integer.  This expression makes it
inviting to set the ratio on the right-hand side to unity by fixing
\begin{equation}
m_{\rm max} = \alpha - 1,
\end{equation}
and thus $U/L = 2$ is, in a sense, the natural first option to
explore.

\begin{assump}
\label{assump:upper-and-lower}
The upper and lower bounds in the fundamental inequalities are related by
\begin{equation}
U = 2L.
\end{equation}
\end{assump}

This lets us solve for $N$ in terms of~$\alpha$:
\begin{equation}
N = (\alpha - 1)^2.
\end{equation}
Thanks to our two latest assumptions, we can fix all three parameters
in the generalized urgleichung (\ref{eq:gen-urgleichung}) in terms of
the maximal size of an MMD set:
\be 
N = m_{\rm max}^2,\ \alpha =
m_{\rm max} + 1,\ \beta = \frac{1}{m_{\rm max}}.
\ee
Relabeling $m_{\rm max}$ by $d$ for brevity, we recover the formulas
familiar from the SIC representation of quantum state space.  Here,
the generalized urgleichung takes the specific form
\be 
q(j) = \sum_{i}\left[(d+1)p(i)
  - \frac{1}{d}\right] r(j|i),
\ee
and we arrive at the following pair of
inequalities:
\begin{equation}
\frac{1}{d(d+1)} \leq \inprod{p}{s}
 \leq \frac{2}{d(d+1)}.
\label{eq:germ-defining}
\end{equation}

Consequently, the polar of a set $A$ is
\be
\dl{A} = \left\{ u \in H \colon \inprod{u}{v} \ge \frac{1}{d(d+1)} \ \forall v \in A\right\}.
\label{eq:our-def-polar}
\ee

We now arrive at the definition upon which the rest of our theory will stand.

\begin{definition}
A \emph{qplex} is a self-polar subset of the out-ball in the
probability simplex $\ps_{d^2}$, with the parameters in the
generalized urgleichung set to $\alpha = (d+1)$ and $\beta = 1/d$.
\end{definition}

\section{Fundamental Geometry of Qplexes}
\label{sec:polarity}
In the previous section, we began with the urgleichung and, making a
few assumptions of an operational character, arrived at the double
inequality
\begin{equation}
\frac{1}{d(d+1)} \leq \inprod{p}{s}
 \leq \frac{2}{d(d+1)}.
\label{eq:double-inequality-repeat}
\end{equation}
Here, we will take this as established, and we will demonstrate
several important geometrical properties of the sets that maximally
satisfy it---the qplexes.

A \qplex\ is a subset of $\ps$, the probability simplex in $\fd{R}^{d^2}$ (i.e. the space of probability distributions with $d^2$ outcomes).  $\ps$ is, in turn, a subset of the hyperplane
\be
H =\left\{u \in \fd{R}^{d^2} \colon  \inprod{u}{c} = \frac{1}{d^2}\right\},
\ee
where $\inprod{\cdot}{\cdot}$ denotes the usual scalar product on $\fd{R}^{d^2}$ and
\be
c = \bmt \frac{1}{d^2}  & \dots & \frac{1}{d^2}\emt^{\rm{T}}
\ee
is the barycenter of $\ps$.

It is important to appreciate the geometrical relationships between the four sets $\ps, \bs, \bo, \bi$.  Specializing our results from the previous section, we have
\ea{
e_i - c &= \frac{1}{d+1}(v_i - c),
\\
\ri & = \frac{1}{d-1} \ro .
}
So the basis simplex is obtained from the probability simplex by scaling by a factor $1/(d+1)$, while the in-ball is obtained from the out-ball by scaling by a factor $1/(d-1)$.  In particular $\bi=\bo$ when $d=2$, but is otherwise strictly smaller.
We have
\be
 \inprod{e_j}{e_k} = \frac{d\delta_{jk} + d + 2}{d(d+1)^2}.
\label{eq:basis-purity-special}
\ee

If $d=2$ then
\be
\bs \subseteq \bi=\bo \subseteq \ps.
\ee
If $d>2$ then  one still has
\be
\bs \cup \bi \subseteq \ps \cap \bo
\ee
but
\ea{
\bs &\nsubseteq  \bi, & \bo \nsubseteq \ps.
}
The first of these statements is an immediate consequence of the foregoing.  To prove the second observe that  $e_i\in \bs$ but $\notin \bi$, while  $c + (\ro/\ri)(\bar{e}_i-c) \in \bo$ but $\notin \ps$.

These facts are perhaps most easily appreciated by examining the diagram in Fig.~\ref{figSimpsAndBalls}.  Observe, however, that the metric relations are impossible to reproduce in a 2-dimensional diagram. So, although Fig.~\ref{figSimpsAndBalls}  reproduces the inclusion relations, and points of contact, it badly misrepresents the sizes of the sets $\bs,\bi$ in comparison to the sets $\ps,\bo$.
\begin{figure}[hbt]
\includegraphics[width = 8cm]{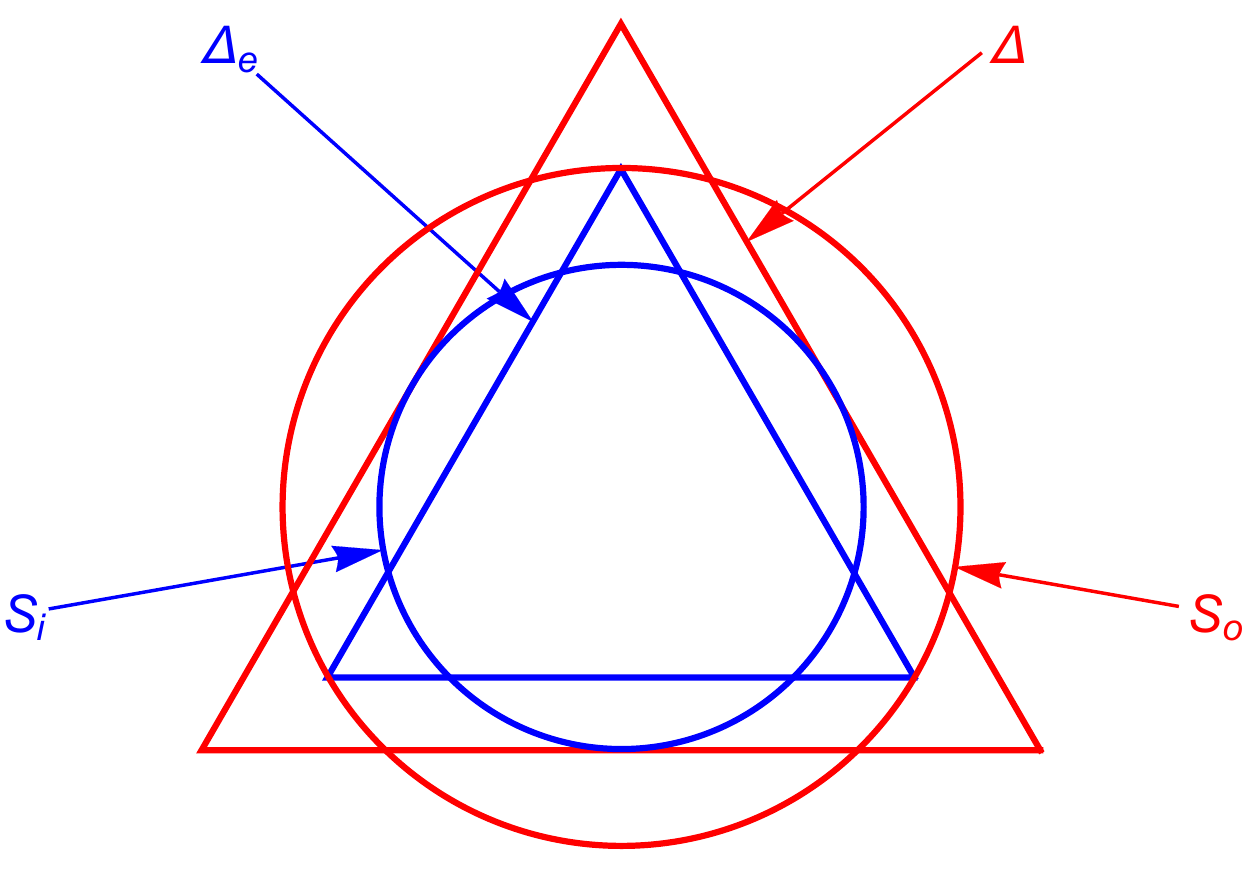}
\caption{\label{figSimpsAndBalls}Relations between the sets $\ps$, $\bs$, $\so$, $\si$ when $d>2$.  The diagram is schematic only.  It shows the inclusion relations, and points of contact, but does not reproduce the metric relations, which are impossible to depict accurately in a 2-dimensional diagram.  In particular, the basis simplex $\bs$ is much smaller in relation to the probability simplex $\ps$ than is shown here.  This is also true of the in-sphere $\si$ and the out-sphere $\so$. }
\end{figure}

General properties of qplexes include the following:
\begin{itemize}
\item Any qplex is convex and closed, and is thus the convex hull of
  its extremal points.

\item Because a qplex is self-polar, it can be thought of as the
  intersection of half-spaces.  Each half-space is defined, per
  Eq.~(\ref{eq:our-def-polar}), by a hyperplane that is composed of
  points all maximally distant from an extreme point of the qplex.

\item For every extreme point of a qplex, there exists at least one
  point that is maximally distant to it, in the sense of saturating
  the lower bound in Eq.~(\ref{eq:double-inequality-repeat}).

\item Call a vector ${p} \in \qp$ a \emph{pure} vector if
  $\inprod{p}{p} = 2/(d(d+1))$.  Any set of pure vectors that pairwise
  saturate the lower bound of the consistency condition
  (\ref{eq:double-inequality-repeat}) contains no more than $d$
  elements.

\item Suppose we have a qplex $\qp$ that is a polytope, \emph{i.e.,} the
  convex hull of a finite set of vertices.  Because all qplexes
  contain the basis distributions, this polytope must have at least
  $d^2$ vertices.  The polar of each extreme point is a half-space
  bounded by a hyperplane, all of the points on which are maximally
  distant from that extreme point.  The intersection of the
  half-spaces defined by all these hyperplanes forms a polytope.  By
  self-polarity, this polytope is identical to~$\qp$.  It follows that
  each extreme point of~$\qp$ must lie on at least $d^2 - 1$ such
  hyperplanes.  Therefore, each vertex of~$\qp$ is maximally distant
  from at least $d^2 - 1$ other extreme points.

\item It follows from the above that a qplex cannot be a simplex.
  Consequently, any point in the interior of a qplex can be written in
  more than one way as a convex combination of points on the boundary.
  This is a generalization of the result that any mixed quantum state
  has multiple convex decompositions into different sets of pure
  states, a theorem that has historically been of some significance in
  interpreting the quantum formalism~\cite{Jaynes1957, Ochs1981,
    CFS2001, stacey-VN}.  Also, a result of Pl\'avala implies that any
  qplex admits incompatible measurements~\cite{Plavala2016}.

\item If $\qp$ is a qplex, then no vector $p \in \qp$ can have an element
  whose value exceeds $1/d$.

\item The total number of zero-valued entries in any vector belonging
  to a qplex is bounded above by $d(d-1)/2$.

\end{itemize}

A SIC representation of a quantum state space is a qplex with a
continuous set of pure points.  All qplexes with this property enjoy
an interesting geometrical relation with the polytopes that can be
inscribed within them.

\begin{theorem}
If $\qp$ is a \qplex\ that contains an infinite number
  of pure points, then any polytope inscribed in~$\qp$ cannot contain
  the in-sphere $\si$.
\end{theorem}
\begin{proof}
Suppose that $P$ is a polytope inscribed in $\qp$ that contains the
in-sphere $\si$.  Recall that the polarity operation reverses
inclusion (Theorem~\ref{tm:polarity}), so the polar polytope $\dl{P}$
of~$P$ must contain the polar $\dl{\qp}$ of~$\qp$.  But all qplexes are
self-polar, so $\qp \subset \dl{P}$.  Likewise, because the polar of the
in-ball $\bi$ is the out-ball $\bo$, it follows that $\dl{P}$ is
contained within the out-sphere $\so$.  Consequently, $\qp$ can have
only a finite number of pure points.
\end{proof}

Let us consider the two-outcome measurement $r_{{s}}$ defined by
rescaling a state ${s} \in \qp$:
\begin{equation}
r_{{s}}(0|i) = d^2 \gamma_0 s(i),\ \gamma_0 = \frac{1}{d}.
\end{equation}
We fix the other row of the matrix $r_{{s}}(j|i)$ by
normalization:
\begin{equation}
r(0|i) + r(1|i) = 1.
\end{equation}
Does this actually define a legitimate measurement?  Because
$\inprod{p}{s}$ is always bounded above and below for any vector ${p}
\in \cP$, then applying $r_{{s}}$ to any ${p} \in \cP$ via the
urgleichung will yield a valid probability vector ${q}$.  Therefore,
$r_{{s}}$ defined in this way is indeed a member of~$\cR$.

What's more, if we apply $r_{{s}}$ to the state ${s}$ itself,
then we can be \emph{certain} about the outcome, if ${s}$ lies on
the same sphere as the basis distributions.  In such a case, we have
$q(0) = 1$.  If Alice ascribes a state having this magnitude to a
system, she is asserting her confidence that performing a particular
experiment will have a specific result.  But certainty about
\emph{one} experiment does not, and indeed cannot, imply certainty
about \emph{all.}  Even when Alice is certain about what would happen
should she perform the experiment $r_{{s}}$, she is necessarily
uncertain about what would happen if she brought the Bureau of
Standards measurement down to the ground and applied it.

Note that when we apply $r_{{s}}$ to a state ${p}$, we compute
\begin{equation}
q(0) = d(d+1)\inprod{p}{s} - 1.
\label{eq:p-s-dual}
\end{equation}
The bound established by Assumption~\ref{assump:upper-bound} implies
that we can associate the factor $d$ just as well with ${s}$ or
with ${p}$.  That is, both $r_{{s}}$ and $r_{{p}}$ are
valid measurements within $\cR$, and we obtain the same probability
$q(0)$ when we apply $r_{{s}}$ to~${p}$ as we would if we
applied $r_{{p}}$ to the state ${s}$.

This is a point worth considering in depth.  With
Assumption~\ref{assump:R-from-P}, we introduced a relation between the
set of all states and the set of all measurements.  Now, thanks to the
additional assumptions we have invoked since then, we have a more
specific correspondence between the two sets:  For every pure state,
there is a binary measurement for which that state, and no other
state, implies certainty.  This result depends upon our assumption
that departures from complete ignorance are minimally constrained, or
equivalently, that the basis distributions are extremal.  As a
consequence, we know that we can take any valid state $s$ and scale by
a factor $d$ to create a row in a measurement matrix.  In the language
of Asher Peres, the fact that we can interpret Eq.~(\ref{eq:p-s-dual})
as $r_s$ applied to~$p$ or as $r_p$ applied to~$s$, for any states $p$
and $s$, is the reciprocity of ``preparations'' and ``tests''~\cite{PeresBook}.

This reciprocity is an important concept for many mathematical treatments
of quantum physics.  For example, it is one of the primary axioms in
Haag's formulation~\cite{Haag, Araki}.  To those who apply category
theory to quantum mechanics, it is the reason why they construct
``dagger-categories,'' and how the basic idea of an inner product is
introduced into their diagrammatic language~\cite{coecke2016}.

Next, we consider sets which are related to qplexes.

\begin{definition}
A subset $A$ of the probability simplex $\ps$ is a \emph{germ} if it
satisfies the fundamental inequalities (\ref{eq:germ-defining}) for
all $p$, $s\in A$.
\end{definition}

\begin{definition}
A germ is \emph{maximal} if no point can be added to it without
violating the fundamental inequalities (\ref{eq:germ-defining}).
\end{definition}

We start by proving two results about germs that follow from the
Cauchy--Schwarz inequality.  Originally, these theorems were proved
for qplexes~\cite{fuchs2009, appleby2011, fuchs2013}, but they apply
more broadly.

\begin{theorem}
If $\pr$ is a germ, then no vector $p \in \pr$ can have an element
  whose value exceeds $1/d$.
\end{theorem}
\begin{proof}
Let $p \in \pr$ be a point on the out-sphere.  Assume without loss of generality that $p(0) \geq p(i)$.  Then
\begin{equation}
\frac{2}{d(d+1)} = p(0)^2 + \sum_{i=1}^{d^2-1} p(i)^2,
\end{equation}
and using the Cauchy--Schwarz inequality,
\begin{equation}
\frac{2}{d(d+1)} \geq p(0)^2 + \frac{1}{d^2-1} \left(\sum_{i=1}^{d^2-1} p(i)\right)^2.
\end{equation}
By normalization, we can simplify the sum in the last term, yielding
\begin{equation}
\frac{2}{d(d+1)} \geq p(0)^2 + \frac{1}{d^2-1} \left(1 - p(0)\right)^2.
\end{equation}
Thus,
\begin{equation}
p(0) \leq \frac{1}{d},
\end{equation}
with equality if and only if all the other $p(i)$ are equal, in which case, normalization forces them to take the value $1/(d(d+1))$.
\end{proof}
\begin{remark}
If the germ $\pr$ contains the basis distributions, this result also follows from
\begin{equation}
\inprod{p}{e_k} = \frac{1}{d(d+1)} + \frac{p_k}{d+1} \leq \frac{2}{d(d+1)}.
\end{equation}
\end{remark}

\begin{theorem}
The total number of zero-valued entries in any vector belonging
  to a germ is bounded above by $d(d-1)/2$.
\end{theorem}
\begin{proof}
Let $\pr$ be a germ and choose $p \in \pr$.  Square the basic normalization condition to find
\begin{equation}
\left( \sum_i p(i) \right)^2 = 1.
\end{equation}
Apply the Cauchy--Schwarz inequality to show, writing $n_0$
for the number of zero-valued elements in $p$,
\begin{equation}
(d^2 - n_0) \sum_{\{i:p(i) >0\}} p(i)^2 \geq
 \left( \sum_{\{i:p(i) >0\}} p(i) \right)^2 = 1.
\end{equation}
Consequently,
\begin{equation}
n_0 \leq d^2 - \frac{d(d+1)}{2} = \frac{d(d-1)}{2}.
\label{eq:weak-zeros-bound}
\end{equation}
\end{proof}

It follows from Zorn's lemma~\cite{Zorn} that every germ is contained
in at least one maximal germ.  In other words, we can extend any germ
in at least one way to form a set that is also a germ, but which
admits no further consistent extension.  Adding any new point to a
maximal germ implies that some pair of points will violate the
inequalities (\ref{eq:germ-defining}).  Every qplex is a germ, but the
converse is not true.  Using the theory of polarity, we will show that
any maximal germ is a self-polar subset of the out-ball.  That is, a
maximal germ is a qplex, and in fact, any qplex is also a maximal
germ.

It is an immediate consequence of the definition that if $\pr$ is an arbitrary germ then
\be
\pr \subseteq \ps \cap \bo,
\label{eq:gmInclude}
\ee
where $\bo$ is the out-ball:
\ea{
\bo &= \left\{u \in H \colon \inprod{u}{u} \le \frac{2}{d(d+1)}
\right\} \\
 & = \left\{ u\in H \colon \| u- c\| \le  \ro \right\}.
}

Taking polars on both sides of Eq.~(\ref{eq:gmInclude}) and taking account of what polarity does to inclusion and intersection (Theorem~\ref{tm:polarity}), we find
\be
\cc (\dl{\ps} \cup \dl{\bo}) \subseteq \dl{\pr}
\label{eq:gmIncludeStar}
\ee
for every germ $\pr$.  Recall from Lemma~\ref{lm:polar-basis} that the polar of~$\Delta$ is the basis simplex $\bs$, and by Lemma~\ref{lm:simpballpolars} we know that the polar of the out-ball $\bo$ is the in-ball $\bi$.  Therefore,
\be
\cc (\bs \cup \bi) \subseteq \dl{\pr}.
\ee

We are now able to prove
\begin{theorem}
\label{tm:qplexPolarity}
Let $A$ be a subset of $\ps\cap \bo$. Then
\begin{enumerate}
\item $A$ is a germ if and only if $A \subseteq \dl{A}$.
\item $A$ is a maximal germ if and only if $A = \dl{A}$.
\end{enumerate}
Therefore, the terms ``maximal germ'' and ``qplex'' are equivalent.
\end{theorem}
\begin{proof}
The first statement is an immediate consequence of the definition.  To
prove the second statement we need to do a little work.  This is
because it is not immediately apparent that if $A$ is a maximal germ then
$\dl{A}\subseteq \bo$.

Suppose that $A$ is a maximal germ.  We know from the first part of the theorem that $A \subseteq \dl{A}$.  To prove the reverse inclusion  let $u\in \dl{A}$ be arbitrary.
In order to show that $u\in A$ first consider the vector
\be
\tilde{u} = c - \frac{\ri}{\|u-c\|} (u-c).
\ee
We have, for all $v\in A$,
\ea{
-\ri \ro \le \inprod{\tilde{u}-c}{v-c}  \le \ri\ro,
\\
\intertext{implying}
\frac{1}{d(d+1)} \le \inprod{\tilde{u}}{v} \le \frac{2}{d(d+1)}.
}
Also
\ea{
\frac{1}{d(d+1)} \le \inprod{\tilde{u}}{\tilde{u}} = \frac{1}{d^2-1} \le \frac{2}{d(d+1)}.
}
So $\tilde{u} \in A$.  We now use this to show that $u \in A$.  In fact
\ea{
-\ri \|u-c\| &= \inprod{u-c}{\tilde{u}-c}
\nonumber
\\
&= \inprod{u}{\tilde{u}} - \frac{1}{d^2}
\nonumber
\\
& \ge
-\ri \ro
}
implying
$u  \in \bo$.
Consequently
\ea{
\la u,v \ra &= \inprod{u-c}{v-c} + \frac{1}{d^2}
\nonumber
\\
&\le \ro^2 + \frac{1}{d^2}
\nonumber
\\
& = \frac{2}{d(d+1)}
}
for all $v\in A$.  The fact that $u\in \dl{A}$ means
\be
\inprod{u}{v} \ge \frac{1}{d(d+1)}
\ee
for all $v \in A$.  Finally
\be
\frac{1}{d(d+1)} \le   \|u-c\|^2 +\frac{1}{d^2} =   \inprod{u}{u} \le \frac{2}{d(d+1)}.
\ee
So $u\in A$.  This completes the proof that if $A$ is a maximal germ then $A=\dl{A}$.  The  converse statement, that if $A=\dl{A}$ then $A$ is a maximal germ, is an immediate consequence of the definition.
\end{proof}

Let us note that this theorem means, in particular, that every maximal
germ contains the basis simplex.  If we start with the fundamental
inequalities and assume that the set of points satisfying them is
maximal, then that set turns out to be self-polar.  Because the state
space is contained within the probability simplex, the state space
must \emph{contain} the \emph{polar of} the probability simplex, which
by Lemma~\ref{lm:polar-basis} is the basis simplex.  In earlier
papers on germs~\cite{fuchs2009, appleby2011, fuchs2013}, the existence of the
basis distributions was an extra assumption in addition to maximality; here,
using the concept of polarity, we have been able to derive it.

Let us also note, as another consequence of this theorem, that if $\qp$ is a maximal germ, and if $q$ is any element of $\qp$, then there exists a measurement $r$ and index $a$ such that $q=s_a$, where $s_a$ is the distribution
\begin{equation}
s_a(j) = \frac{r(a|j)}{\sum_k r(a|k)}.
\end{equation}
This too was something that was assumed in older work~\cite{fuchs2009, appleby2011, fuchs2013}, but which we are now in a position to derive. To see that it is true observe that the statement is trivial if $q=c$ (simply take $r$ to be the one-outcome measurement).  If, on the other hand,   $q\neq c$ we can define
\be
q' =  c-\frac{\ri}{\|q-c\|}(q-c).
\ee
By construction $q' \in \si$.  So it follows from the theorem that $q'\in \qp$.  Consequently, if we define
\ea{
r(a|i)
=
\begin{cases}
\frac{d^2\ri}{\|q-c\|+\ri} q(i) \qquad & a = 1
\\
\frac{d^2\|q-c\|}{\|q-c\|+\ri} q'(i) \qquad & a=2
\end{cases}
}
then $r$ describes a two-outcome measurement such that
\ea{
\frac{r(1|i)}{\sum_k r(1|k)} &= q(i), & \frac{r(2|i)}{\sum_k r(2|k)} &= q'(i).
}

At this stage, we turn to the question of what germs can have in common, and how they can differ.  In order to develop this topic, we introduce some more
definitions.  Given an arbitrary germ $\pr$, let $\allqplexes_\pr$ denote the set of all \qplexes\ containing $\pr$ (necessarily nonempty, as we noted above).

\begin{definition}
The \emph{stem} of a germ $\pr$ is the set
\begin{equation}
\stm(\pr) = \bigcap_{\qp \in \allqplexes_\pr} \qp,
\end{equation}
and the \emph{envelope} of $\pr$ is the set
\begin{equation}
\env(\pr) = \bigcup_{\qp \in \allqplexes_\pr} \qp.
\end{equation}
\end{definition}
When $\pr$ is the empty set  $\allqplexes_{\emptyset}$ is the set of all \qplexes, without restriction.  In that case we omit the subscript and simply denote it $\allqplexes$.  Similarly we write $\stm(\emptyset) = \stm$ and $\env(\emptyset) = \env$.  We will refer to $\stm$ and $\env$ as the principal \stem\ and \envelope.

\begin{theorem}
Let $\pr$ be a germ.    Then
\ea{
\stm(\pr) &= \cc(\bs \cup \bi \cup \pr),
\label{eq:trnExpn}
\\
\env(\pr) &= \ps \cap \bo \cap \dl{\pr}.
\label{eq:envExpn}
}
In particular, $\stm(\pr)$ and $\env(\pr)$ are mutually polar.
\end{theorem}
\begin{proof}
If $\qp$ is a qplex containing $\pr$ we must have
$\qp =\dl{\qp} \subseteq \ps \cap \bo \cap \dl{\pr}$.  So
\be
\env(\pr) \subseteq \ps\cap \bo \cap \dl{\pr}.
\ee
On the other hand if $p$ is any point in  $\ps \cap \bo \cap \dl{\pr}$ then $\pr \cup \{p\}$ is a germ, and so must be contained in some $\qp\in \allqplexes_\pr$.  The second statement now follows.

To prove the first statement we take duals on both sides of
\be
\bigcup_{\qp \in \allqplexes_\pr} \qp = \ps \cap \bo \cap \dl{\pr}.
\ee
We find
\ea{
\stm (\pr)& = \cc\Bigl(\bs \cup \bi \cup \cc\bigl(\pr\cup \{c\}\bigr)\Bigr)
\nonumber
\\
&= \cc(\bs \cup \bi \cup \pr).
}
\end{proof}
\begin{corollary}
\label{cor:mainTrEnv}
The principal \stem\ and \envelope\ are given by
\ea{
\stm &= \cc(\bs \cup \bi),
\\
\env &= \ps \cap \bo.
}
\end{corollary}
\begin{proof}
Immediate.
\end{proof}
This result is illustrated schematically in Figure~\ref{figCoreEnvelope}.

\begin{figure}[hbt]
\includegraphics[width = 8cm]{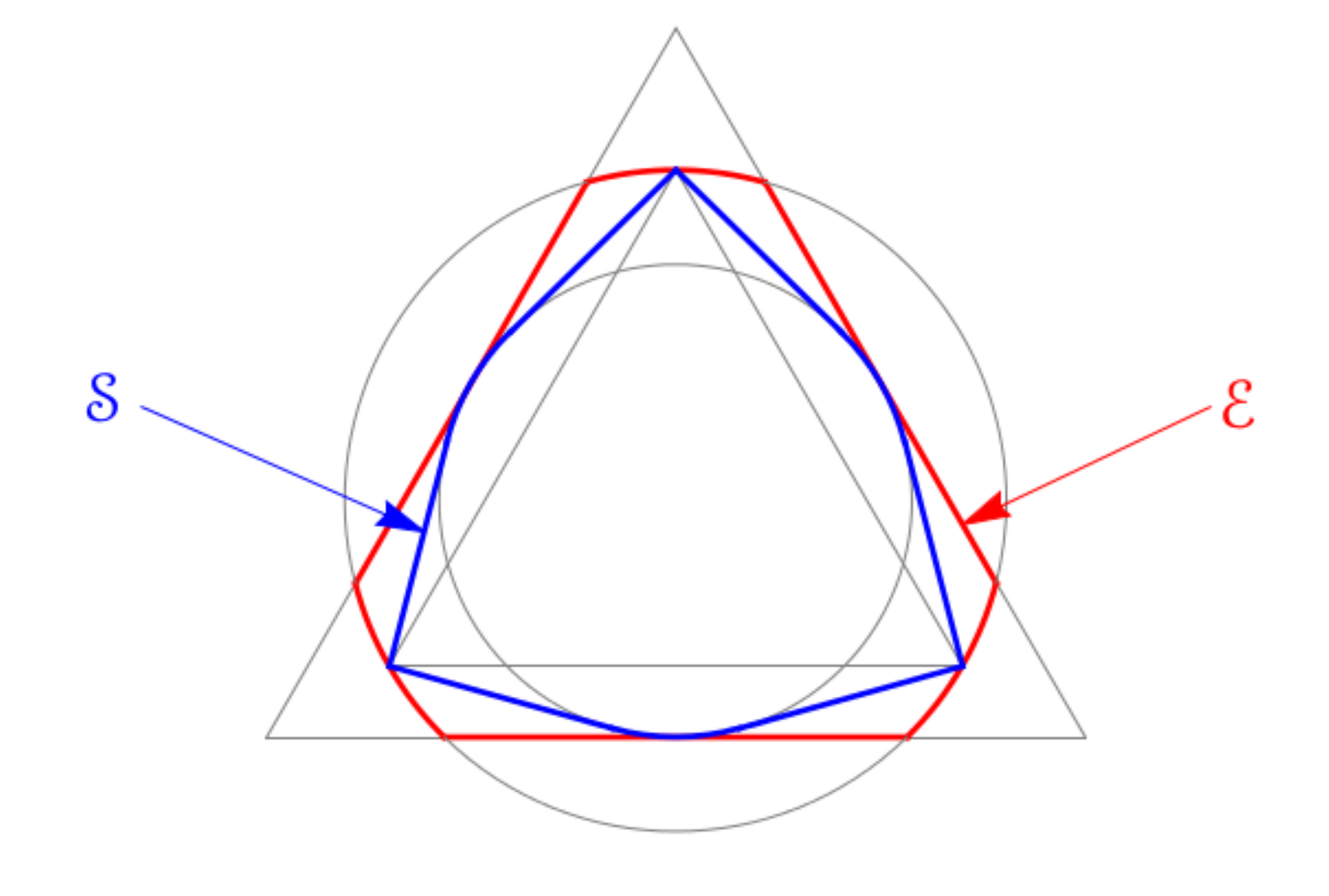}
\caption{\label{figCoreEnvelope}The principal stem and envelope when $d>2$.  The sets $\ps$, $\bs$, $\so$, $\si$ are shown in  gray.  The surface of every qplex lies between the blue and red surfaces.  As with Fig.~\ref{figSimpsAndBalls} the diagram is schematic only.  In particular, the total measure of the set $\stm$ is much smaller in comparison  with the total measure of the set $\env$ than the diagram suggests.}
\end{figure}

\begin{corollary}
\label{cr:germExtend}
Let $\pr$ be a closed, convex germ containing $\stm$.  Then
\begin{align}
\stm(\pr) &  = \pr & \env(\pr) & = \dl{\pr}
\end{align}
Moreover, given arbitrary $p\in \dl{\pr}$ such that $p \notin \pr$ there exist \qplexes\ $\qp_1$, $\qp_2$ containing $\pr$ such that
\begin{align}
p &\notin \qp_1 & p& \in \qp_2
\end{align}
\end{corollary}
\begin{proof}
Immediate.
\end{proof}
Every germ can be extended to a \qplex.  It is natural to ask how many ways there are of performing the extension.  The following theorem provides a partial answer to that question.
\begin{theorem}
Let $\pr$ be a closed, convex germ containing $\stm$.  If $\pr$ is not already a \qplex, then there are uncountably many \qplexes\ containing $\pr$.
\end{theorem}
\begin{proof}
It will be convenient to begin by introducing some notation.  Given any two  points $p_1$, $p_2\in H$ we define
\ea{
[p_1,p_2]&=\{\lambda p_1 + (1-\lambda)p_2 \colon 0 \le \lambda \le 1\}
\\
(p_1,p_2]&=\{\lambda p_1 + (1-\lambda)p_2 \colon 0 < \lambda \le 1\}
\\
[p_1,p_2)&=\{\lambda p_1 + (1-\lambda)p_2 \colon 0 \le \lambda < 1\}
\\
(p_1,p_2)&=\{\lambda p_1 + (1-\lambda)p_2 \colon 0 < \lambda <1\}
}

Turning to the proof, suppose that
 $\pr$ is not a \qplex.  Then we can choose $p \in \dl{\pr}$ such that $p \notin \pr$.  Let $q$ be the point where $[c,p]$ meets the boundary of $\pr$.  We will show that, for each  $s \in (q,p)$ there exists a \qplex\ $\qp_s$ such that $\pr \cup [c,s] \subseteq \qp_s$ and $(s,p] \cap \qp_s = \emptyset$.  The result will then follow since $\qp_s \neq \qp_{s'}$ if $s\neq s'$.

To construct the \qplex\ $\qp_s$ for given $s\in (q,p)$ observe that it follows from the basic theory of convex sets~\cite{grun} that there exists a hyperplane through $s$ and not intersecting $\pr$.   This means we can choose  $u\in H$ such that
\ea{
\la u, v \ra & <   \la u, s\ra   =1
}
for all $v\in\pr$. Observe that for all $t \in (s,p)$ we have
\ea{
t = \lambda s + (1-\lambda) c
}
for some $\lambda >1$ and, consequently,
\ea{
\la u, t \ra =
\frac{(d^2-1)\lambda  + 1}{d^2} > 1.
}
Let
\ea{
u' = \left(1 + \frac{1}{(d+1)(d^2-1)}\right) c - \frac{1}{(d+1)(d^2-1)} u
}
and let $A=\cc\bigl( \pr \cup \{s\}\bigr)$.  Then it is easily seen that $u'\in \dl{A}$ while
\ea{
\la u' , t \ra & < \frac{1}{d(d+1)}
}
for all $t \in (s,p)$.  $A$ is a closed, convex germ containing $\stm$, so it follows from Corollary~\ref{cr:germExtend} that there exists a \qplex\ $\qp_s$ containing $A$ and $u'$.  By construction $t \notin \qp_s$ for all $t\in (s,p)$, so $\qp_s$ has the  required properties.
\end{proof}
The result just proved shows that there exist uncountably many \qplexes.  However, we would like to know a little more:  namely, how many \qplexes\ there are which are geometrically distinct.  We now prove a series of results leading to Theorem~\ref{thm:infNonIsomorphic}, which states that there are uncountably many \qplexes\ which are not isomorphic to each other, or to quantum state space.

\begin{definition}
Let $s\in H$ be arbitrary.  We define the polar point of $s$ to be the point
\ea{
\dl{s} &= c - \frac{\ro\ri}{\|s-c\|^2} (s-c)
\\
\intertext{and the polar hyperplane of $s$ to be the set}
H_s &=\left\{u \in H \colon \la u, s\ra = \frac{1}{d(d+1)}\right\}.
}
\end{definition}Observe that $\ddl{s} = s$ for all $s$, and
\ea{
\la \dl{s} , s\ra = \frac{1}{d(d+1)}
}
(so $\dl{s} \in H_s$).
The relations between the polar $\dl{\{s\}}$, the polar point $\dl{s}$ and the polar hyperplane $H_s$ are depicted in Fig.~\ref{fig:PolarPoint}.
\begin{figure}[hbt]
\includegraphics[width = 6cm]{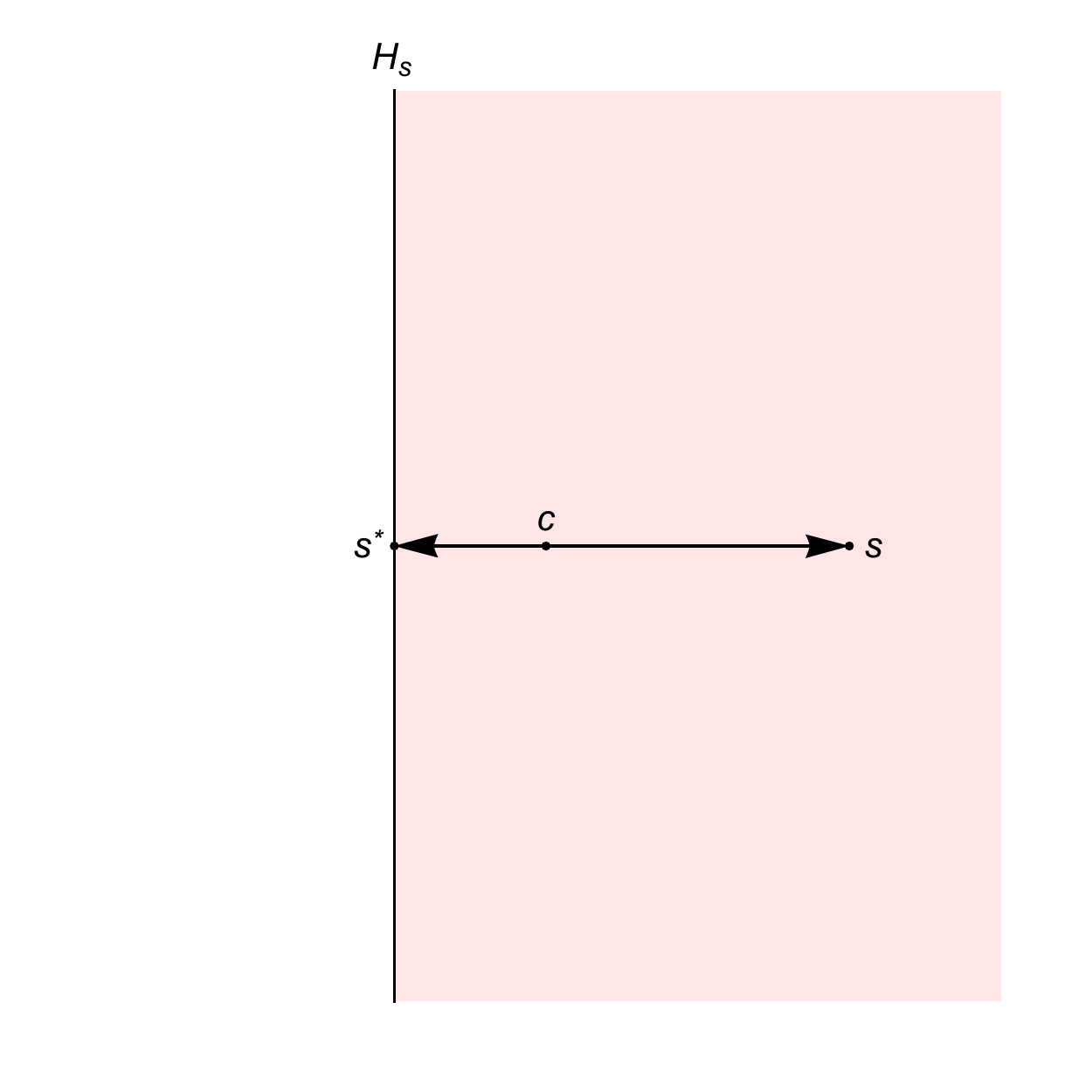}
\caption{\label{fig:PolarPoint} Diagram to illustrate the relationships between the polar $\dl{\{s\}}$, the polar point $\dl{s}$ and the polar hyperplane $H_s$.  $\dl{\{s\}}$ is the pink-shaded region to the right of $H_s$.}
\end{figure}
It follows from these definitions that if $s$ is any point on $\so$ (respectively $\si$), then $\dl{s}$ is on $\si$ (respectively $\so$).

\begin{theorem}
\label{thm:dualOfPtOnSo}
Let $\pr$ be a closed germ, and let $s$ be any point on $\so$.  Then $s\in \pr$ if and only if $\dl{s}$  is on the boundary of $\dl{\pr}$.
\end{theorem}
\begin{remark}
Specializing to the case when $\pr$ is a \qplex, the theorem says that the points where the boundary of $\pr$ touches the out-sphere are antipodal to the points where it touches the in-sphere.  This is a subtle property of quantum state space~\cite{Kimura2005, Appleby2007}.
\end{remark}
\begin{proof}
Suppose $s\in \pr$.  Then it follows that $\dl{s} \in \si$.  The fact that $\pr$ is a germ means $\pr \subseteq\env$, implying $\stm \subseteq \dl{\pr}$.  So $\dl{s} \in \dl{\pr}$.  Moreover, if we define
\ea{
t_n = \frac{n+1}{n} \dl{s} -\frac{1}{n} c,
}
then
\ea{
\la t_n , s\ra
&= \frac{1}{d(d+1)} - \frac{1}{nd^2(d+1)} < \frac{1}{d(d+1)}
}
for all $n$.  So $t_n$ is a sequence outside $\dl{\pr}$ converging to $\dl{s}$.  We conclude that $\dl{s}$ is on the boundary of $\dl{\pr}$.

Conversely, suppose $\dl{s}$ is on the boundary of $\dl{\pr}$.  Then we can choose a sequence $t_n \notin \dl{\pr}$ such that $t_n \to \dl{s}$.  For each $n$ there must exist $p_n\in \pr$ such that
\ea{
\la t_n , p_n \ra < \frac{1}{d(d+1)}.
}
Since $\pr$ is closed and bounded it is compact.  A theorem of point
set topology has it that in a compact set, every sequence contains a
convergent subsequence.  Therefore, we can choose a convergent
subsequence $p_{n_j} \to p\in \pr$.  Also, the fact that $t_{n_j} \to
\dl{s}$ means $\dl{t}_{n_j} \to s$.  So
\be
\| p - s\|^2 = \lim_j \left( \| p_{n_j} - \dl{t}_{n_j}\|^2 \right).
\ee
We can expand the quantity inside the limit as
\be
\| p_{n_j} - \dl{t}_{n_j}\|^2
 =
\|p_{n_j}-c\|^2 + \|\dl{t}_{n_j}-c\|^2 - 2 \inprod{\dl{t}_{n_j}-c}{p_{n_j}-c}.
\ee
In turn, we have that
\be
\lim_j \left( \| p_{n_j} - \dl{t}_{n_j}\|^2 \right)
 \leq
\|p-c\|^2 - \ro^2.
\ee
Because $p$ is contained in the out-ball, its distance from $c$ has to be less than $\ro$, meaning that
\be
\| p - s\|^2 \leq \|p-c\|^2 - \ro^2 \leq 0.
\ee
So $p$ coincides with $s$, which consequently belongs to $\pr$.
\end{proof}
We next prove two results which show that we can restrict our attention to the out-sphere when trying to establish the existence of non-isomorphic \qplexes.  The first of these, Lemma~\ref{lm:extendSoconst}, is a technical result which will also be used in Section~\ref{sec:qgroups}.
\begin{lemma}
\label{lm:extendSoconst}
Let $\pr$ be a closed germ containing $\stm$, and let
\ea{
C = \pr \cup (\dl{\pr}\cap \midball)
}
where $\midball$ is the mid-ball.  Then $C$ is a closed germ such that  $C'\cap \so = \pr \cap \so$ for every germ $C'$ containing $C$.
\end{lemma}
\begin{proof}
It follows from Theorem~\ref{tm:polarity} and Lemma~\ref{lm:simpballpolars} that
\ea{
\dl{C} &=  \dl{\pr} \cap \cc(\pr \cup \midball),
}
from which one  sees that $C\subseteq \dl{C}$. Moreover, the fact that $\stm \subseteq \pr$ means $\dl{\pr} \subseteq \ps \cap \bo$, implying $C\subseteq \ps\cap\bo$.  So $C$ is a closed germ containing $\pr$.  Let $C'$ be any germ containing $C$.   It is immediate that $\pr\cap \so = C \cap \so \subseteq C'\cap\so$.  Suppose, on the other hand, that $s$ is a point on $\so$ not belonging to $\pr$.  Let $\dl{\pr}_{\rm{b}}$,  $C_{\rm{b}}$ be the boundaries of $\dl{\pr}$, $C$ respectively.  The fact that $C\cap \midball = \dl{\pr}\cap \midball$ is easily seen to imply  $\dl{\pr}_{\rm{b}}\cap \si = C_{\rm{b}}\cap\si$.  So  it follows from  Theorem~\ref{thm:dualOfPtOnSo} that $\dl{s}\notin C_{\rm{b}}$. Since $\si \subseteq C$ this means  $\dl{s}$ must lie in the interior of $C$.  So there exists $\lambda > 1$ such that
\be
t = \lambda \dl{s} + (1-\lambda) s
\ee
is in $C$.  Since
\ea{
\la t, s \ra  =\frac{2-\lambda}{d(d+1)} < \frac{1}{d(d+1)},
}
it follows that  $s \notin \dl{C}$.  Consequently $s\notin C'$.
\end{proof}
\begin{theorem}
\label{thm:bdryTheoremA}
Let $\pr$ be a closed germ containing the vertices of the basis simplex.  Then there exists a \qplex\ $\qp$ such that $\qp\cap \so=\pr\cap \so$.
\end{theorem}
\begin{proof}
Let
\ea{
\tilde{\pr} = \cc(\pr\cup \bs \cup \bi).
}
Then $\tilde{\pr}$ is a closed germ containing $\stm$.  Moreover $\tilde{\pr}\cap \so = \pr \cap \so$.  Let
\ea{
C &= \tilde{\pr} \cup \left(\dl{\tilde{\pr}}\cap \midball\right).
}
Then it follows from  Lemma~\ref{lm:extendSoconst} that $C$ is a germ and that $\qp\cap \so = \pr\cap \so$ for any \qplex\ $\qp$ containing $C$.
\end{proof}
Before proving Theorem~\ref{thm:infNonIsomorphic} we need to give a sharp definition of what it means for two \qplexes\ to be isomorphic.
\begin{definition}
\label{def:qplexIsomorphism}
We say that two \qplexes\ $\qp$ and $\qp'$ are isomorphic if and only if there exists a linear bijection $f\colon \fd{R}^{d^2} \to \fd{R}^{d^2}$ such that
\begin{enumerate}
\item $\qp'=f(\qp)$.
\item For all $q_1$, $q_2 \in \qp$
\ea{
\la f(q_1),f(q_2) \ra = \la q_1, q_2\ra.
}
\end{enumerate}
\end{definition}

We are now ready to prove the final result of this section.
\begin{theorem}
\label{thm:infNonIsomorphic}
There exist uncountably many \qplexes\ which are not isomorphic to each other.
\end{theorem}
\begin{proof}
We have
\ea{
\la e_i , e_j\ra = \frac{d\delta_{ij} + d + 2}{d(d+1)^2}
}
for all $i$, $j$ (c.f. Eq.~(\ref{eq:basis-purity-special})).  So if we define
\ea{
p_{\theta}& = c + \cos \theta (e_1-c) + \sin \theta (e_2-c)
\\
\pr_{\theta} &= \{p_{\theta} , e_1, \dots, e_{d^2}\}
}
then, for sufficiently small $\epsilon$,  the set $\pr_{\theta}$ is a germ for all $\theta \in (0,\epsilon)$.  It follows from Theorem~\ref{thm:bdryTheoremA} that we can choose \qplexes\ $\qp_{\theta}$ such that $\qp_{\theta} \cap \so = \pr_{\theta} \cap \so = \pr_{\theta}$.  By construction, the scalar products $\inprod{p_{\theta}}{e_j}$ are different for different choices of~$\theta$, and so \qplexes\ $\qp_{\theta}$ corresponding to different values of $\theta$ are non-isomorphic.  Moreover, the fact that the intersection with $\so$ is finite means that $\qp_{\theta}$ is non-isomorphic to quantum state space for all $\theta$.
\end{proof}
So far we have been focussing on  \qplexes\ in general.  However, it seems to us that the method of analysis employed is a potentially insightful way of thinking about the geometry of quantum state space.

\section{Type-preserving measurements}
\label{sec:intg}
We now come to the central result of this paper.   We will show that the symmetry group of a \qplex\ can be identified with a set of measurements, which in turn can be identified with a set of regular simplices within the qplex whose vertices all lie on the out-sphere.

Let $\qp$ be a \qplex\ and $r$ a measurement with $n$ outcomes.  For each $q\in \qp$, let $q_r$ be the distribution given by the urgleichung, Eq.~(\ref{eq:SICMeasProbs}).  Then the map $q \to q_r$ takes $\qp$ to
\be
\qp_r = \{q_r \colon q\in \qp\} \subseteq \Delta_n
\ee
(where  $\Delta_n$ is the $n-1$ dimensional probability simplex). We refer to $\qp_r$ as the measurement set, and the map $q\to q_r$ as the measurement map.

We are interested in  measurements having $d^2$ outcomes for which the measurement set is another \qplex.  We will refer to such measurements as type-preserving.  We are particularly interested in the case when the measurement set is $\qp$ itself, in which case we will say that the measurement is $\qp$-preserving.

Let $r$ be an arbitrary measurement.  Then it is easily seen that the urgleichung can  be written in the alternative  form
\be
q_r(i) = \sum_j R_{ij} q(j),
\ee
where
\be
R_{ij} = (d+1) r(i|j) -\frac{1}{d} \sum_k r(i|k).
\label{eq:RmtDef}
\ee
We refer to $R$ as the stretched measurement matrix. Note that Eq.~(\ref{eq:RmtDef}) can be inverted:
\be
r(i|j) = \frac{1}{d+1} \left( R_{ij} + \frac{1}{d} \sum_k R_{ik}\right).
\ee
So the stretched measurement matrix uniquely specifies the measurement.

Now specialize to the case of a type-preserving measurement.  In that case it turns out that   $R$ must be an orthogonal matrix.  To see this we begin by observing that, since the basis simplex belongs to both $\qp$ and $\qp_r$, there must exist  $s_i \in \qp$, $s'_i\in \qp_r$ such that
\ea{
Rs_i &= e_i, &  Re_i &= s'_i.
\label{eq:sspdef}
}
We then have
\begin{lemma}
\label{lm:typPMeasDetREq1}
Let $R$, $s_i$, $s'_i$ be as above.  Then
\begin{enumerate}
\item $\det R  = \pm 1$.
\item $s_i$, $s'_i\in \so$ for all $i$.
\item $\cc(\{s_i\})$ and $\cc(\{s'_i\})$ are regular simplices.
\end{enumerate}
\end{lemma}
\begin{proof}
The proof is based on the fact~\cite{VolSimp} that the simplices of maximal volume within a ball are precisely the regular simplices with vertices on the sphere that bounds the ball.  The desired result follows from considering the simplex formed by the $s_i$ and the origin (and the corresponding simplex formed by the $s'_i$ and the origin).
\end{proof}

To complete the proof that $R$ is an orthogonal matrix, we observe that maps from regular simplices to regular simplices are orthogonal.  From this, we can derive the following theorem.
\begin{theorem}
\label{thm:TPreserveMapTheorem1}
Let $\qp$ be a \qplex, and let $R$ be the stretched measurement matrix of a type-preserving measurement.  Then $R$ is an orthogonal matrix such that $Rc=c$.   Moreover there exists a regular simplex with vertices $s_i \in \qp\cap\so$ such that
\ea{
R_{ij} &= (d+1)s_i(j) -\frac{1}{d},
\label{eq:MeasMatTermsSimp}
\\
Rs_i &= e_i,
\\
(Re_i)(j) &= s_j(i).
}
\end{theorem}
\begin{remark}
We will refer to  $\cc(\{s_i\})$ as the measurement simplex.
\end{remark}

For a given \qplex\ $\qp$ define
\begin{enumerate}
\item $\typ{\qp}$ to be the class of type-preserving measurements.
\item $\sym{\qp}$ to be the class of regular simplices with vertices in $\qp\cap \so$.
\item $\ort{\qp}$ to be the class of orthogonal matrices $R$ such that $R\qp$ is a \qplex.
\end{enumerate}
The previous theorem states that to each element of $\typ{\qp}$ there corresponds an element of $\sym{\qp}$ and an element of $\ort{\qp}$.  The next theorem we prove states that the  correspondences are in fact bijective, so that we can identify the three classes $\typ{\qp}$, $\sym{\qp}$ and $\ort{\qp}$.
\begin{theorem}
\label{thm:RegSimpIsTPreserve}
Let $\qp$ be a $\qplex$ and let $s_i\in \qp\cap \so$ be the vertices of a regular simplex $\Delta_s$.  Then  $\Delta_s$ is the measurement simplex of a type-preserving measurement.  Likewise, if $R$ is an orthogonal matrix such that $R\qp$ is also a qplex, then $R$ is the stretched measurement matrix for a type-preserving measurement.
\end{theorem}
\begin{proof}
Define
\ea{
r(i|j) = s_i(j).
}
It is immediate that the $r(i|j)$ are the conditional probabilities defining a measurement with stretched measurement matrix
\be
R_{ij} = (d+1)s_i(j) -\frac{1}{d}.
\ee
We need to show that the measurement is type-preserving.  In other words, we need to show that the set $R\qp$ is a qplex.  For all $q\in \qp$
\ea{
(Rq)(i) &= (d+1) \la s_i , q\ra - \frac{1}{d},
}
from which it follows
\ea{
(Rq)(i) &\ge 0 & \sum_i (R q)(i) &= 1.
}
So $R\qp\subseteq \ps$.  Also,
it follows from the same considerations that led to Theorem~\ref{thm:TPreserveMapTheorem1} that $R$ is orthogonal.  The defining condition of a germ, Eq.~(\ref{eq:germ-defining}), is invariant under orthogonal transformations.  Therefore, $R\qp$ is a qplex.

We now prove the other direction of the correspondence.  Let $R$ be an orthogonal matrix such that $R\qp$ is a qplex.  We know that the basis distributions $e_i$ must belong to $R\qp$.  So, there exist $s_i\in \qp$ such that
\be
e_i = Rs_i.
\label{eq:biRsi}
\ee
Since $\det R = \pm 1$ we have, by the same argument used to prove Lemma~\ref{lm:typPMeasDetREq1}, that the $s_i\in \so$ and are the vertices of a regular simplex.  It now follows from the considerations above that $\cc(\{ s_i\})$ is the measurement simplex of a type-preserving measurement, with stretched measurement matrix
\be
R'_{ij} = (d+1) s_i(j) -\frac{1}{d}.
\ee
By multiplying both sides of Eq.~(\ref{eq:biRsi}) by $R^{\rm T}$, we find that
\ea{
s_i(j) &= \sum_k R^{\rm{T}}_{jk} e_i(k)
\nonumber
\\
&=\frac{1}{d+1} R_{ij} + \frac{1}{d(d+1)} \sum_{k} R_{kj}.
\label{eq:MeasMatTermsSimpCalc}
}
Summing over $i$ on both sides of this equation we find
$
\sum_k R_{kj} = 1
$
for all $j$ and, consequently, $R=R'$.
\end{proof}

At this stage, we recall our definition of an isomorphism between
qplexes: Two \qplexes\ $\qp$ and $\qp'$ are isomorphic if and only if
there exists an inner-product-preserving map $f\colon \fd{R}^{d^2} \to
\fd{R}^{d^2}$ that sends $\qp$ to $\qp'$.

\begin{theorem}
Let $\qp$, $\qp'$ be \qplexes.  Then $\qp$ and $\qp'$ are isomorphic if and only if there is a type-preserving measurement on $\qp$ such that $\qp' = R\qp$, where $R$ is the stretched measurement matrix.
\end{theorem}
\begin{proof}
Sufficiency is immediate.  To prove necessity suppose that $f\colon \qp \to \qp'$ is an isomorphism.  The fact that $f$ preserves scalar products on a set which spans $\fd{R}^{d^2}$ means that it must be represented by an orthogonal matrix.  The claim now follows from Theorem~\ref{thm:RegSimpIsTPreserve}.
\end{proof}
So far we have been looking at type-preserving measurements in general.  Let us now focus on the special case of $\qp$-preserving measurements.  Suppose that we have two such measurements, with measurement matrices $R$, $R'$.  Then $RR'$ is also an orthogonal matrix, with the property that $RR'\qp = \qp$.  So it follows from Theorem~\ref{thm:RegSimpIsTPreserve} that $RR'$ is the stretched measurement matrix for a $\qp$-preserving measurement.  Similarly with $R^{\rm{T}}$, the inverse.  In short,  the $\qp$-preserving measurement maps form a group.  For ease of reference let us give it a name:
\begin{definition}
Let $\qp$ be a \qplex.  The preservation group of $\qp$, denoted $\qgp{\qp}$, is the group of type-preserving measurement maps between $\qp$ and itself.
\end{definition}

The elements of $\qgp{\qp}$ are symmetries of  $\qp$.  The question naturally arises, whether they comprise \emph{all} the symmetries.  The above considerations are not sufficient to answer that question because they leave open the possibility that $\qp$ is invariant under  orthogonal transformations which do not fix the origin of $\fd{R}^{d^2}$.  The following theorem eliminates that possibility.
\begin{theorem}
\label{tm:preservation-symmetry}
Let $\qp$ be a \qplex.  Then the preservation group  is the symmetry group of $\qp$.
\end{theorem}
\begin{proof}
The symmetry group of a subset of a normed vector space is defined to be the group of isometries of the set.  It has been shown above that every $\qp$-preserving measurement map is an isometry of $\qp$. We need to show the converse.  Let $f$ be an isometry of $\qp$.  It follows from Theorem~\ref{thm:TPreserveMapTheorem1} that $f(c) = c$.

Now define a map $\tilde{f}\colon \qp-c \to \qp-c$ by
\ea{
\tilde{f}(u) = f(u+c) - c
}
One easily sees that $\|\tilde{f}(u)\| = \|u\|$ for all $u \in \qp-c$.  Consequently
\be
\tilde{f}(u) = T u
\ee
for some orthogonal transformation $T$ of the subspace $H-c$.  We may extend $T$ to an orthogonal transformation $R$ of the whole space $\fd{R}^{d^2}$ by defining $Rc = c$.  It is then immediate that $R\qp = \qp$.  The result now follows by Theorem~\ref{thm:RegSimpIsTPreserve}.

\end{proof}

\section{From preservation group to \qplex}
\label{sec:qgroups}
In this section we ask what conditions a subgroup of $\Ot(d^2)$ must
satisfy in order to be the preservation group of some \qplex.  This
will lead us to the question of when symmetries are powerful enough to
determine a qplex essentially uniquely.  Let $\qp$ be a qplex and $\gp$ be
its preservation group.  Under what conditions can $\qp$ be maximally
symmetric, in the sense that $\gp$ is not a proper subgroup of the
symmetry group of any qplex?  The answer will turn out to depend upon
how the group $\gp$ acts on the basis simplex.

Quantum state space has the property that any pure state can be mapped
to any other pure state by some unitary operation, that is, by some
symmetry of the state space.  Indeed, given any pure state, the set of
all pure states is the orbit of the original state under the action of
the symmetry group.  This leads us to consider the general question of
qplexes whose extremal points form a single orbit under the action of
the qplex's symmetries.  One can prove that if $\qp$ is such a qplex,
then the symmetry group of~$\qp$ is maximal, and furthermore, any other
qplex $\qp'$ with the same symmetry group is identical to~$\qp$.

Given a group $\gp \subseteq \Ot(d^2)$, can $\gp$ be the preservation group
of a qplex? It is easy to find a necessary condition.  Following our
previous paper \cite{GroupAlg}, we introduce the concept of a
stochastic subgroup:
\begin{definition}
A subgroup $\gp\subseteq \Ot(d^2)$ is stochastic if, for all $R\in \gp$,
\be
R_{ij} \ge -\frac{1}{d} \quad \forall i,j
\quad\quad \hbox{ and }  \quad\quad
R c = c.
\ee
\end{definition}
Equivalently, we may say that a subgroup $\gp\subseteq \Ot(d^2)$ is stochastic if every matrix in $\gp$ is of the form
\be
R_{ij} = (d+1) S_{ij} - \frac{1}{d},
\ee
where $S_{ij} = s_i(j)$ is a doubly-stochastic matrix (hence the name).
It can  then be seen from Theorem~\ref{thm:TPreserveMapTheorem1} that every preservation group is a stochastic subgroup of $\Ot(d^2)$.

It is natural to ask whether the condition is sufficient as well as necessary, so that every stochastic subgroup of $\Ot(d^2)$  is the preservation group of some \qplex.  We have not been able to answer this question in full generality.  However, we have obtained some partial results.  We can show that any stochastic subgroup $\gp \subseteq \Ot(d^2)$ is at least contained in the preservation group of some qplex.  To see why, we start with a preliminary result.
\begin{lemma}
\label{thm:assocgerm}
Let $\gp$ be a stochastic subgroup of $\Ot(d^2)$.  For each $R\in \gp$ define the vectors $s^R_i$ by applying $R$ to the basis distributions:
\ea{
s^R_i(j) = \frac{1}{d(d+1)}\left( dR_{ij} + 1\right).
}
Then  $s^R_i\in \ps\cap \so$ for all $i$ and $\cc(\{s^R_i\})$ is a regular simplex.  Moreover
\ea{
\pr = \{s^R_i \colon R \in \gp, \ i = 1, \dots, d^2\}
}
is a germ.
\end{lemma}
\begin{proof}
Straightforward consequence of the definitions.
\end{proof}
\begin{definition}
Let $\gp$ be a stochastic subgroup of $\Ot(d^2)$.  The orbital germ
is the orbit of the basis distributions under the action of~$\gp$, that
is, the set $\pr$ specified in the statement of
Lemma~\ref{thm:assocgerm}.
\end{definition}
\begin{theorem}
\label{thm:StochSgpContQgp}
Let $\gp$ be a stochastic subgroup of $\Ot(d^2)$.  Then there exists a \qplex\ $\qp$ such that $\gp \subseteq \qgp{\qp}$.
\end{theorem}
\begin{proof}
Let $\pr$ be the orbital germ of $\gp$, and let $\allgerms_\pr$ be the set of all germs $P$ such that
\begin{enumerate}
\item $P$ contains $\pr$.
\item $RP=P$ for all $R \in \gp$.
\end{enumerate}
It follows from Zorn's lemma that $\allgerms_\pr$ contains at least one maximal element.  Let $\qp$ be such a maximal element.  Observe that if $P$ is in $\allgerms_\pr$ then its convex closure is also in $\allgerms_\pr$; consequently $\qp$ must be convex and closed.  Observe, also, that if $R$ is any element of $\gp$, then $c$ is in the interior of the simplex $\cc(\{s^R_i\})$; consequently $c$ is in the interior of $\qp$.

 We claim that $\qp$ is in fact a \qplex.  For suppose it were not.  Then we could choose $p \in \ps\cap \bo \cap \dl{\qp}$ such that $p\notin \qp$.  For each $\lambda$ in the closed interval $[0,1]$ define $p_{\lambda}=\lambda p + (1-\lambda) c$.  The fact that $\qp$ is closed, convex together with the fact that $c$ is in the interior of $\qp$ means that there exists $\lambda_0 \in (0,1)$ such that  $p_{\lambda} \in \qp$ if and only if $\lambda \in [0,\lambda_0]$.  We have
\ea{
\la p , Rp_{\lambda} \ra \ge \frac{1}{d(d+1)}
}
for all $R \in \gp$, $\lambda \in [0,\lambda_0]$. Consequently
\ea{
\la  p_{\lambda} , R p_{\lambda} \ra  \ge \frac{1}{d(d+1)} + \frac{1-\lambda}{d^2(d+1)}
}
for all $\lambda \in [0,\lambda_0]$, $R\in \gp$.  By continuity this inequality must hold for all $\lambda \in [0,\lambda_0]$, $R\in \bar{\gp}$, where $\bar{\gp}$ is the closure of $\gp$ in $\Ot(d^2)$.   It follows that there must exist a fixed number $\mu \in (\lambda_0, 1]$ such that
\ea{
\la   p_{\mu} , R p_{\mu}\ra \ge \frac{1}{d(d+1)}
}
for all $R \in \gp$.  For suppose that were not the case.  Then we could choose a sequence $\nu_n \downarrow \lambda_0$, and a sequence $R_n \in \gp$, such that
\ea{
\la  p_{\nu_n} , R_n p_{\nu_n}\ra < \frac{1}{d(d+1)}
}
for all $n$.  The group $\bar{\gp}$ is compact (because $\Ot(d^2)$ is compact~\cite{CompactGroup}) as is the closed interval $\left[0,\frac{1}{d(d+1)}\right]$.  Consequently we can choose a subsequence $n_j$ such that $R_{n_j} \to \bar{R}\in \bar{\gp}$ and
\ea{
\la p_{\nu_{n_j}} , R_{n_j} p_{\nu_{n_j}}\ra \to a
}
for some $a \in \left[0,\frac{1}{d(d+1)}\right]$.  But this would imply that
\ea{
\la p_{\lambda_0} , \bar{R} p_{\lambda_0} \ra  = a \le \frac{1}{d(d+1)}
}
---which is a contradiction.

Now consider the set
\ea{
\qp' = \qp \cup \{R p_{\mu} \colon R \in \gp\}.
}
Observe that \ea{
(Rp_{\mu})(i) &= (d+1)\la s^R_i , p_{\mu}\ra - \frac{1}{d} \ge 0
}
for all $i$ and all $R\in \gp$ (because $p_{\mu} \in \dl{\qp}\subseteq \dl{\pr}$).  So $\qp' \subseteq \ps$.  It is immediate that $\qp'\subseteq \bo$ and $\qp'\subseteq \qp^{\prime *}$.  So $\qp'$ is a germ such that $R\qp' = \qp'$ for all $R\in \gp$, and which is strictly larger than $\qp$---which is a contradiction.

It is now immediate that $\gp$ is a subgroup of $\qgp{\qp}$.
\end{proof}
We can make  stronger statements if we introduce some new concepts.
\begin{definition}
A stochastic subgroup $\gp\subseteq \Ot(d^2)$ is maximal if it is not contained in any larger stochastic subgroup.
\end{definition}
\begin{definition}
A stochastic subgroup $\gp\subseteq \Ot(d^2)$ is strongly maximal if it is maximal and if, in addition, the closed convex hull of the orbital germ is a \qplex.
\end{definition}
We then have the following results.
\begin{corollary}
\label{cor:maxStochQplex}
Let $\gp$ be a maximal stochastic subgroup of $\Ot(d^2)$.  Then there exists a \qplex\ $\qp$ such that $\gp=\qgp{\qp}$.
\end{corollary}
\begin{proof}
Immediate consequence of Theorem~\ref{thm:StochSgpContQgp}.
\end{proof}
\begin{theorem}
Let $\gp$ be a strongly maximal stochastic subgroup of $\Ot(d^2)$ and let $\pr$ be the orbital germ.  Then $\cc(\pr)$ is the unique \qplex\ $\qp$ such that $\gp=\qgp{\qp}$.
\end{theorem}
\begin{proof}
We know from Corollary~\ref{cor:maxStochQplex} that there exists at least one \qplex\ $\qp$ such that $\gp=\qgp{\qp}$.  If $\qp$, $\qp'$ are \qplexes\ such that $\gp=\qgp{\qp} = \qgp{\qp'}$ then $\qp$, $\qp'$ must both contain $\cc(\pr)$, where $\pr$ is the orbital germ.  Since $\cc(\pr)$ is a qplex we must have $\qp = \cc(\pr) = \qp'$.
\end{proof}

This brings us back to the claim we made at the beginning of this section.

\begin{corollary}
If $\qp$ is a qplex whose extreme points form a single orbit under the action of the preservation group, then the preservation group of~$\qp$ is strongly maximal.
\end{corollary}
\begin{proof}
Let $\qp$ be a qplex and $\gp$ be its preservation group.  Assume that the extremal points form a single orbit under the action of~$\gp$.  The basis distributions are among the extremal points, so all extremal points are on the same orbit as any basis distribution.  In other words, the orbital germ is the set of extreme points.  Suppose that $\qp'$ is a qplex whose preservation group contains $\gp$.  Then $\qp'$ contains all the extremal points of~$\qp$, and thus, $\qp'$ contains $\qp$.  But a qplex is a maximal germ, so we must have $\qp' = \qp$.
\end{proof}

 \section{Characterizing \qplexes\ isomorphic to quantum state space}
 \label{sec:character}
We are, of course, most interested in \qplexes\ corresponding to SIC measurements.  In this section, we will define what it means for a \qplex\ to be isomorphic to quantum state space.  We will prove that if $\qp$ is a \qplex\ isomorphic to quantum state space, then its preservation group is isomorphic to the projective extended unitary group, essentially the group of all unitaries and anti-unitaries with phase factors quotiented out.  Then, we will establish the converse: If the preservation group of a qplex is isomorphic to the projective extended unitary group, then that qplex is isomorphic to quantum state space.  This result indicates one way of recovering quantum theory from the urgleichung.

\begin{definition}
\label{def:qstpIsomorphism}
Let $\herm$ be the space of Hermitian operators on $d$-dimensional Hilbert space  and let $\dens$ be the space of density matrices.  We will say that a \qplex\ $\qp$ is isomorphic to quantum state space if there exists an $\fd{R}$-linear bijection $f\colon \herm \to \fd{R}^{d^2}$  such that
\begin{enumerate}
\item $\qp = f(\dens)$.
\item For all $\rho$, $\rho'\in \dens$
\ea{
\bigl< f(\rho), f(\rho')\bigr>=\frac{ \Tr(\rho \rho') +1 }{d(d+1)}.
\label{eq:qstIsoProp2}
}
\end{enumerate}
A qplex that is isomorphic to quantum state space will be designated a Hilbert qplex.
\end{definition}
It is straightforward to verify that definitions~\ref{def:qplexIsomorphism} and~\ref{def:qstpIsomorphism} are consistent, in the sense that if $\qp$ is a Hilbert qplex, and if $\qp'$ is any other \qplex, then $\qp'$ is a Hilbert qplex if and only if it is isomorphic to $\qp$ in the sense of definition~\ref{def:qplexIsomorphism}.
\begin{theorem}
Let $\qp$ be a \qplex.  Then a map $f\colon \dens \to \qp$ is an isomorphism of quantum state space onto $\qp$ if and only if there is a SIC $\Pi_j$ such that
\ea{
(f(\rho))(j) = \frac{1}{d} \Tr(\rho \Pi_j)
\label{eq:frhojTermsSIC}
}
for all $j$ and all $\rho\in \dens$.
\end{theorem}
\begin{remark}
Thus, to each isomorphism of quantum state space onto $\qp$, there corresponds a unique SIC.  In particular a SIC exists in dimension $d$ if and only if a Hilbert qplex exists in dimension $d$.
\end{remark}
\begin{proof}
Suppose $f\colon S\to \qp$ is an isomorphism.  Define
\be
\Pi_j = f^{-1}(e_j).
\ee
Then
\ea{
\Tr(\Pi_j\Pi_k) & = d(d+1)\la e_j, e_k \ra -1
=\frac{d\delta_{jk} + 1}{d+1}.
}
So $\Pi_j$ is a SIC.  Moreover, for all $\rho\in \dens$, and all $j$,
\ea{
\frac{1}{d} \Tr(\rho \Pi_j) &= (d+1) \la f(\rho), e_j\ra -\frac{1}{d}
= (f(\rho))(j).
}
Suppose, on the other hand, $f\colon S\to \qp$ is a map for which Eq.~(\ref{eq:frhojTermsSIC}) is satisfied for some SIC $\Pi_j$.  Then we can extend $f$ to a linear bijection of $\herm$ onto $\fd{R}^{d^2}$.  We know from prior work~\cite{fuchs2009, appleby2011, fuchs2013} that  $f(S)$ is a \qplex.  Since it is contained in $\qp$ we must have $f(S) = \qp$.  Moreover, since
\ea{
\rho = \sum_j \left( (d+1) (f(\rho))(j) -\frac{1}{d}\right) \Pi_j,
}
with a similar expression for $\rho'$, we have
\ea{
\Tr(\rho \rho') = d(d+1) \la f(\rho), f(\rho') \ra -1
}
from which Eq.~(\ref{eq:qstIsoProp2}) follows.
\end{proof}

One might wonder if other \qplexes, not isomorphic to $\qp$ (and we know that these exist, per Theorem~\ref{thm:infNonIsomorphic} and Appendix~\ref{sec:altQplex}), correspond to other informationally complete POVMs.  This is not the case.  It follows from the foregoing that there is no measurement which will take us from a \qplex\ of one kind to a \qplex\ of a different, nonisomorphic kind.

Knowing this, let us characterize the preservation group of a Hilbert \qplex\ $\qp$.  We define the extended unitary group, denoted $\EU(d)$,  to be the group consisting of all unitary and anti-unitary operators, and the projective extended unitary group, denoted $\PEU(d)$, to be the quotient $\EU(d)/\M(d)$, where  $\M(d)$ is the sub-group consisting of all unitaries of the form $e^{i\theta} I$, for some phase $e^{i\theta}$.
\begin{theorem}
\label{eq:qisoSGpPeU}
Let $\qp$ be a Hilbert qplex.  Then $\qgp{\qp}$ is isomorphic to $\PEU(d)$.
\end{theorem}
\begin{proof}
Straightforward consequence of Wigner's theorem~\cite{Wigner}.
\end{proof}

 We showed in  Theorem~\ref{eq:qisoSGpPeU} that if $\qp$ is a Hilbert qplex then $\qgp{\qp}$ is isomorphic to $\PEU(d)$.  Now, we will prove the converse:  If $\qgp{\qp}$ is isomorphic to $\PEU(d)$, then $\qp$ is a Hilbert qplex.  It turns out, in fact, that a weaker statement is true:  If $\qgp{\qp}$ contains a subgroup isomorphic to $\PU(d)$, then $\qp$ is a Hilbert qplex.

 In the Introduction we remarked on the need for an extra assumption,
 additional to the basic definition of a qplex, which will serve to
 uniquely pick out those \qplexes\ which correspond to quantum state
 space.  The theorem we will prove momentarily supplies us with one
 possible choice for this assumption.  As we remarked in the
 introduction, there may be others.

 As a by-product of this result we obtain a criterion for SIC existence:  Namely, a SIC exists in dimension $d$ if and only if $\PU(d)$ is isomorphic to a stochastic subgroup of $\Ot(d^2)$.  We proved this result by another method in a previous paper~\cite{GroupAlg}, but this is the route by which we were originally led to it. Indeed, it is hard to see why it should occur to anyone that stochastic subgroups of $\Ot(d^2)$ might be relevant to SIC existence if they were not aware of the role that such subgroups play in the theory of \qplexes.

The result depends on the following method for embedding a \qplex\ in operator space.  The question of whether a SIC exists in every dimension is very hard, and, indeed, is still unsolved.  But if one simply asks for a set of operators $\Pi_1, \dots , \Pi_{d^2}$ satisfying the equations
\ea{
\Tr(\Pi_j) &= 1,
\label{eq:quasiSICDef1}
\\
\Tr(\Pi_j \Pi_k) & = \frac{d\delta_{jk} +1}{d+1},
\label{eq:quasiSICDef2}
}
without imposing any further constraint---in particular,
without requiring that the $\Pi_j$ be positive semi-definite---then the problem becomes almost trivial.  To see this consider the real Lie algebra $\su(d)$ (i.e.\ the space of trace-zero Hermitian operators).  Equipped with the Hilbert--Schmidt inner product
\be
\la B, B'\ra  = \Tr(B B'),
\ee
this becomes a $(d^2-1)$-Euclidean space, so the existence of operators $B_1,\dots, B_{d^2}$, each of length 1, and forming the vertices of a regular simplex, is guaranteed.  These operators satisfy
\be
\Tr(B_j B_k)
=
\begin{cases}
1 \qquad & j=k;\\
-\frac{1}{d^2-1} \qquad & j\neq k.
\end{cases}
\ee
If we now define
\be
\Pi_j =\sqrt{\frac{d-1}{d}} B_j + \frac{1}{d} I,
\ee
then the $\Pi_j$ satisfy Eqs.~(\ref{eq:quasiSICDef1}) and~(\ref{eq:quasiSICDef2}).  We will refer to them as a quasi-SIC.

Now let $\qp$ be an arbitrary \qplex, and for each $q\in \qp$ define, by analogy with Eq.~(\ref{eq:rhoTermsProbs})
\ea{
\rho_q = \sum_j \left( (d+1) q(j) -\frac{1}{d}\right) \Pi_j.
}
If $\Pi_j$ really were a SIC, and if the $q(j)$ really were the outcome probabilities for a measurement with that SIC, then $\rho_q$ would be a density matrix. In general, however, neither of those conditions need hold true.  So, $\rho_q$ will typically not be positive semi-definite (though it will be trace-$1$).  We will refer to it as a quasi-density matrix.  It will also be convenient to define
\ea{
S_\qp = \{\rho_q\colon q \in \qp\}.
}
We will refer to $S_\qp$ as quasi-state space.
It is easily verified that
\be
0 \le \la \rho, \rho' \ra \le 1,
\label{eq:UrungleichungForQuasiStSp}
\ee
for all $\rho, \rho' \in S_\qp$, just as is the case for genuine density matrices.

We are now in a position to prove
\begin{theorem}
\label{thm:QStateCriterion}
Let $\qp$ be a \qplex.  Then the following statements are equivalent:
\begin{enumerate}
\item $\qgp{\qp}$ contains a subgroup isomorphic to $\PU(d)$.
\item $\qp$ is a Hilbert qplex.
\end{enumerate}
\end{theorem}
\begin{proof}
The implication $(2)\implies (1)$ is an immediate consequence of Theorem~\ref{eq:qisoSGpPeU}.  It remains to prove the implication $(1) \implies (2)$.

Let $\Pi_j$ be a quasi-SIC, and use this quasi-SIC to map the qplex $\qp$ into operator space, creating the quasi-state space $S_\qp$.  The fact that the \qplex{} $\qp$ contains a subgroup isomorphic to the projective unitary group $\PU(d)$ implies that the quasi-state space $S_\qp$ is invariant under unitary transformations.  That is, the projective unitary symmetry of one set carries over to the other.  This result is fairly natural; for completeness, we provide an explicit proof in Appendix~\ref{sec:UnitaryImplication}.

Suppose $q\in \qp\in \so$.  Then
\ea{
\Tr(\rho_q) = \Tr(\rho_q^2) = 1.
}
Also, it follows from Eq.~(\ref{eq:UrungleichungForQuasiStSp}) and unitary invariance of the quasi-state space that
\ea{
0 \le \Tr(\rho_q U\rho_qU^{\dagger}) \le 1.
}
for every unitary $U$.  By choosing $U$ to give the appropriate permutation of the eigenvalues we deduce that
\be
0 \le \sum_i \lambda^{\uparrow}_i \lambda^{\downarrow}_i \le 1,
\ee
where $\lambda^{\uparrow}_i$ (respectively $\lambda^{\downarrow}_i $) are the eigenvalues of $\rho_q$ arranged in increasing (respectively decreasing) order.

We now invoke a lemma proven in~\cite{GroupAlg}.  If $\lambda$ is a vector in $\mathbb{R}^d$ such that
\be
\sum_{j=0}^{d-1} \lambda_j = \sum_{j=0}^{d-1} \lambda_j^2 = 1,
\ee
then
\be
\inprod{\lambda^{\uparrow}}{\lambda^{\downarrow}} \leq 0.
\ee
The inequality is saturated if and only if $d-1$ entries in $\lambda$ are equal.  This can occur when
\be
\lambda^{\downarrow} = (1,0,\ldots,0),
\ee
or when
\be
\lambda^{\downarrow} = \left(\frac{2}{d},\ldots,\frac{2}{d},
 \frac{2}{d} - 1\right).
\ee

So we must have
\be
\sum_i \lambda^{\uparrow}_i \lambda^{\downarrow}_i =0.
\ee
Moreover, the possible solutions for the eigenvalue spectrum $\lambda^{\downarrow}$ imply that either $\rho_q  = P$ or $\rho_q = (2/d)I -P$ for some  rank-$1$ projector $P$.  If $d=2$, then $\rho_q$ is a rank-$1$ projector either way.  Otherwise, if $d>2$, suppose $q$, $q'\in \qp\in \so$ were such that $\rho_q = P$ and $ \rho_{q'} = (2/d)I - P'$  where $P$ and $P'$ are rank-$1$ projectors.  In that case there would be a unitary $U$ such that $UP'U^{\dagger} = P$, which would mean, by unitary invariance, that the quasi-state space contained both $P$ and $(2/d)I -P$.  But
\be
\Tr\left(P \left( \frac{2}{d}I - P\right) \right) = \frac{2}{d} - 1 < 1,
\ee
which contradicts  Eq.~(\ref{eq:UrungleichungForQuasiStSp}).  We conclude that if $d>2$ then, either $\rho_q$ is a rank-$1$ projector for all $q\in \qp$, or else $(2/d)I - \rho_q$ is a rank-$1$ projector for all $q\in \qp$.  In the latter case we may define a new quasi-SIC
\be
\tilde{\Pi}'_j = \frac{2}{d}I - \tilde{\Pi}_j.
\ee
One easily verifies that the new quasi-state space is also unitarily invariant.  Moreover, if we define
\ea{
\rho'_q = \sum_j \left( (d+1) q(j) -\frac{1}{d}\right) \tilde{\Pi}'_j,
}
then
\ea{
\rho'_q = \frac{2}{d} I - \rho_q,
}
implying that $\rho'_q$ is a rank-$1$ projector for all $q\in \qp\in \so$.  There is therefore no loss of generality in assuming that our original quasi-state space is such that $\rho_q$ is a rank-$1$ projector for all $q\in \qp\in\so$. Since
\be
\rho_{e_i} = \tilde{\Pi}_i,
\ee
this means in particular that the $\Pi_i$ are rank-$1$ projectors, and therefore constitute a genuine SIC.

Let us note that unitary invariance means that the set $\{\rho_q \colon q \in \qp\in \so\}$ does not merely consist  of rank-$1$ projectors; it actually comprises all the rank-$1$ projectors.   It follows, that if $\rho$ is an arbitrary density matrix, and if $q(j) = (1/d) \Tr(\rho \tilde{\Pi}_j)$, then $q$ is a convex combination of points in $\qp\in \so$, and therefore $q \in \qp$.  Since the SIC probabilities are a \qplex, it follows that $\qp$ does not contain any other points than these, and is therefore isomorphic to quantum state space as claimed.
\end{proof}

Let us observe that in proving this theorem we have incidentally shown that if there is  a \qplex\ $\qp$ which contains an isomorphic copy of $\PU(d)$, then a SIC exists in dimension $d$.  So the theorem has the following corollary:
\begin{corollary}
The following statements are equivalent:
\begin{enumerate}
\item  $\PU(d)$ is isomorphic to a stochastic subgroup of $\Ot(d^2)$.
\item A SIC exists in dimension $d$.
\end{enumerate}
\end{corollary}
\begin{proof}
The implication $(2)\implies (1)$ is an immediate consequence of Theorem~\ref{thm:QStateCriterion}.  To prove the implication $(1) \implies (2)$, let $\gp$ be a stochastic subgroup of $\Ot(d^2)$ which is isomorphic to $\PU(d)$.  It follows from Theorem~\ref{thm:StochSgpContQgp} that there exists a \qplex\ $\qp$ such that $\gp \subseteq \qgp{\qp}$.  In view of Theorem~\ref{thm:QStateCriterion} this implies $\qp$ is the set of outcome probabilities  for a SIC measurement, which means, in particular, that a SIC must exist in dimension $d$.
\end{proof}

\section{Discussion}
\label{sec:future}
Our investigation of qplexes exists in the context of many years' effort toward the goal of reconstructing quantum theory.  Early pioneers of the subject, like Birkhoff and von Neumann, sought a broader mathematical environment in which quantum theory could be seen to dwell.  This led to the subjects of quantum logic and Jordan algebras~\cite{Mccrimmon}.  However, despite the mathematical developments, the influence on physics---and, indeed, on the philosophy thereof---was rather subdued.  The instensely mathematical character of the work may have played a role in this.  Moreover, this work predated the invention and integration into physics of information theory, which turned out to be a boon to the reconstruction enterprise.  It also predated the theorems of Bell, Kochen and Specker~\cite{mermin1993, mermin1993-erratum}, and thus it could not benefit from their insight into what is robustly strange about quantum physics.

One might say that the ``modern age'' of quantum reconstructions was inaugurated by Rovelli in 1996.  He advocated a research program of deriving quantum theory from physical principles, in a manner analogous to the derivation of special relativity's mathematical formalism~\cite{Rovelli1996}.  During the same time period, one of the authors (CAF) also began advocating this project~\cite{fuchs2002, fuchs2010, transcript}.  An early success was Hardy's ``Quantum theory from five reasonable axioms''~\cite{Hardy01, Schack03}, which pointed out the importance of what we call a Bureau of Standards measurement~\cite[p.\ 368]{Fuchs2014}.

Looking over the papers produced in this ``modern age,'' one technical commonality worth remarking upon is the idea of building up the unitary (or projective unitary) group from a universal gate set~\cite{Masanes, HoehnWever}.  This is an idea from the field of quantum computation.  For example, it is known that any unitary operator can be broken down into a sequence of two-level unitaries, applied in succession~\cite[p.\ 188]{MikeAndIke}.  Also, given a collection of $N$ qubits, all the projective unitaries acting on their joint state space---that is, the group $\PU(2^N)$---can be synthesized using single-qubit unitaries and an entangling gate, like a Controlled \textsc{not} operation, that can be applied to any pair of qubits~\cite{Harrow}.  This suggests one way of making progress in the theory of qplexes, by replacing the unitarity assumption.

Recall that in any qplex, a set of mutually maximally distant points can have at most $d$ elements~\cite{appleby2011, transcript}.  Thus, although a qplex is originally defined as living within a $d^2$-dimensional space, in a sense it has an ``underlying dimensionality''~\cite{transcript} equal to~$d$.  Consider a qplex $\qp$, equipped with a set of $d$ mutually maximally distant pure states.  What if we require that any $d-1$ of those states defines a structure isomorphic to a smaller qplex?  Applying this recursively, we arrive eventually at the condition that any two maximally distant points define a set of probability distributions isomorphic to a qplex with $d = 2$, which is automatically a Hilbert qplex.  This is a strong condition, although it makes no direct mention of a particular symmetry group.  At the moment, we see no way to satisfy this condition other than having $\qp$ be a Hilbert qplex.

Alternatively, one can try to make progress by relaxing the unitarity assumption.  For example, instead of imposing a particular symmetry group, what if we seek the qplexes of maximal allowed symmetry?  Assuming that a SIC exists in dimension $d$, then a qplex in~$\ps_{d^2}$ can be at least as symmetric as a Hilbert qplex.  We conjecture that no qplex can be more symmetric than a Hilbert qplex, where we quantify the degree of symmetry by, for example, the dimension of the Lie group of qplex-preserving maps.  This conjecture leads to another:  We suspect that of all the qplexes of a given dimension, the Hilbert qplexes have maximal Euclidean volume.

Another outstanding question is, out of all the conceivable additions
one could make to probability theory in order to relate expectations
for different hypotheticals, why pick the urgleichung?  To our
knowledge, no one considered such a relation before quantum mechanics
and the SIC representation.  And yet, it is a comparatively mild
modification of the classical relationship.  This is particularly
evident when the measurement on the ground is modeled by a set of $d$
orthogonal projectors, \emph{i.e.,} when it is a von Neumann
measurement.  In that case,
\begin{equation}
q(j) = (d+1) \sum_i p(i) r(j|i) - 1.
\end{equation}
This is just a rescaling and shifting of the classical
formula~\cite{transcript}.

In Section~\ref{sec:basic}, we began with a general affine
relationship between Bureau of Standards probabilities and the
probabilities for other experiments.  By invoking a series of
assumptions, we narrowed the parameter values in the generalized
urgleichung down to those that occur in quantum theory.  (Our last
assumption, which fixed the upper bound at twice the lower bound, may
be related to the choice of complex numbers over real numbers and
quaternions for Hilbert-space coordinates~\cite{fuchs2009}.  For an
unexpected connection between SICs and the normed division algebras,
see~\cite{stacey-sporadic, stacey-hoggar}.)  This has
the appealing feature that a linear stretching is just about
the simplest deformation of the classical Law of Total Probability
that one can imagine.  However, this area is still, to a great extent,
unknown territory:  Why linearity?  Are qualitatively greater
departures from classicality mathematically possible?

Many of the quantum reconstruction efforts to date share the feature
that they make quantum physics as unremarkable as possible: While the
technical steps from axioms to theorems are unassailable, the choice
of axioms gives little insight into what is truly strange about
quantum phenomena.  To borrow a phrase from David Mermin, these
re-expressions tend to make quantum theory sound ``benignly
humdrum''~\cite{mermin-pillow}.

For example, should one aim to derive quantum theory from the fact
that quantum states cannot be cloned?  Arguably not: Even classical
distributions over phase space are uncloneable~\cite{caves1996}.  What
about quantum teleportation?  At root, teleportation is a protocol for
making information about one system instead relevant to another, and
it has exact analogues in classical statistical
theories~\cite{spekkens2007, bartlett2012, spekkens2014}.  In 2003,
Clifton, Bub and Halvorson~\cite{clifton2003} proposed a derivation of quantum theory
that started with $C^*$ algebras and then added, as postulates, some
results of quantum information science, such as the no-broadcasting
theorem~\cite{barnum1996}.  However, the no-broadcasting
theorem---despite its original
motivation~\cite[p.\ 2235]{Fuchs2014}---also applies in classical
statistical theories~\cite{spekkens2007, bartlett2012, spekkens2014},
and thus seems a poor foundation to build the quantum upon.  Overall,
it seems that choosing $C^*$ algebras for a starting point implicitly
does a great deal of the work already~\cite[p.\ 1125]{Fuchs2014}.

Similarly, a more recent derivation by Chiribella, D'Ariano and
Perinotti~\cite{chiribella11} invokes, at a key juncture, the
postulate that any mixed state can be treated as a marginal of a pure
state ascribed to a larger system.  This postulate, the purifiability
of mixed states, is an essential ingredient in their recovery of
quantum theory.  As with the examples above, however, it is also true
in classical statistical theories~\cite{spekkens2007, bartlett2012,
  spekkens2014, disilvestro2016}.  From that perspective, it is
consequently a less than fully compelling candidate for the essence of
quantumness.

By contrast, we have chosen as our starting point what we consider to
be the ``jugular vein'' of quantum strangeness: Theories of intrinsic
hidden variables do so remarkably badly at expressing the vitality of
quantum physics.  The urgleichung is our way of stating this physical
characteristic of the natural world in the language of probability.
Quantum states, it avers, are catalogues of expectations---but
\emph{not} expectations about hidden variables.  This view is in line
with ``participatory realist'' interpretations of quantum
mechanics~\cite{cabello2015, fuchs2016}, like QBism~\cite{fuchs2013,
  Voldemort, FMS-AJP} and related approaches~\cite{zeilinger2005, kofler2010,
  appleby2013}.
  
\section{Acknowledgements}

CAF thanks Ben Schumacher for helpful discussions and the Max Planck
Institute for Quantum Optics for a safe haven in which to work out
some of these radical ideas.  MA acknowledges support by the
Australian Research Council via EQuS project number CE11001013.  We
thank John DeBrota for comments.

\appendix
\section{A \qplex\ which does not correspond to a SIC measurement}
\label{sec:altQplex}

By definition, a qplex is a subset of the probability simplex $\Delta_{d^2}$ such that each pair of points within it satisfy the fundamental inequalities,
\begin{equation}
\frac{1}{d(d+1)} \leq \sum_i p(i)s(i) \leq \frac{2}{d(d+1)}.
\end{equation}
We can construct a qplex which is not isomorphic to quantum state space in the following way.  Begin with a set $A$ defined by the intersection of the probability simplex with the ball
\begin{equation}
\sum_i p(i)^2 \leq \frac{2}{d(d+1)}.
\end{equation}
Our plan is to trim this set down until it becomes a qplex.  First, we break $A$ into $d^2!$ regions, which we label $F_k$, for $k = 1,\ldots,d^2!$.  We define the region $F_1$ to be all probability vectors in the set $A$ whose entries appear in decreasing magnitude.  That is,
\begin{equation}
F_1 = \left\{ {p} : {p} \in A
              \hbox{ and } p(1) \geq p(2) \geq \cdots \geq p(d^2)
      \right\}.
\end{equation}
The region $F_1$ is consistent with the fundamental inequalities, because for every ${p} \in F_1$,
\begin{equation}
\inprod{p}{p} \geq \inprod{p}{c}
 \geq \frac{1}{d^2} > \frac{1}{d(d+1)}.
\end{equation}
We define the other regions $F_k$ analogously.  Because $k$ runs from 1 to $(d^2)!$, it labels the permutations in the symmetric group on $d^2$ elements.  Each $F_k$ consists of the vectors obtained by taking the vectors in $F_1$ and permuting the components according to permutation $k$.  All of the regions $F_k$ so defined will be internally consistent.

To obtain a qplex $\qp$, start with $F_1$ and include all the points from $F_1$ in $\qp$.  Then, take all the points from $F_2$ that are consistent with all the points in $F_1$, and include them in $\qp$.  Continue in this manner, adding the points in each $F_k$ that are consistent with every point added to $\qp$ so far.  The end result will be a qplex that is surely not isomorphic to quantum state space.

\section{Unitary Symmetry of Quasi-State Spaces}
\label{sec:UnitaryImplication}
Let $\Pi_j$ be a quasi-SIC, as defined in Eqs.~(\ref{eq:quasiSICDef1}) and (\ref{eq:quasiSICDef2}) of the main text.  For each $U\in\PU(d)$ we have a matrix $S^U_{jk}$ such that
\ea{
U\Pi_jU^{\dagger} = \sum_k S^U_{jk} \Pi_k.
}
The matrix is given explicitly by
\ea{
S^{U}_{jk} &= \frac{d+1}{d} \Tr\bigl( \Pi_k U \Pi_j U^{\dagger}\bigr) - \frac{1}{d},
}
from which one sees
\ea{
\sum_j S^{U}_{jk} &= 1, & \sum_k S^{U}_{jk} &= 1,
}
and
\ea{
\sum_k S^{U}_{ik}S^{U}_{jk} &= \frac{d+1}{d} \Tr\left(\left( \sum_kS^{U}_{ik}\Pi_k\right)  U\Pi_jU^{\dagger}\right)-\frac{1}{d}
\nonumber
\\
&= \delta_{ij}.
}
So $S^{U}_{ij}$ is an orthogonal matrix.

We now appeal to the assumption that $\qgp{\qp}$ contains a subgroup isomorphic to $\PU(d)$.  So for each $U\in \PU(d)$ there exists an orthogonal matrix $R^{U}_{jk} \in \qgp{\qp}$.  It can be proven that, up to equivalence, the adjoint representation of~$\PU(d)$ is the only nontrivial irreducible representation of~$\PU(d)$ having degree $d^2 - 1$ or smaller, when $d \geq 2$~\cite{GroupAlg}.  Thus, the two representations here must be equivalent, so that
\ea{
R^U = T S^U T^{-1}
\label{eq:equivRepEq}
}
for all $U$ and some fixed orthogonal matrix $T$.  Summing over $k$ on both sides of
\ea{
\sum_{j} R^U_{ij} T_{jk} = \sum_{j}T_{ij}S^U_{jk}
}
and appealing to the fact that the representations are irreducible on the subspace orthogonal to $c$ we deduce that
\ea{
\sum_{j} T_{ij} = t
}
for some constant $t$, independent of $i$.  Similarly
\ea{
\sum_{i} T_{ij} = s
}
for some constant $s$, independent of $i$.  Since
\ea{
d^2 t = \sum_{ij} T_{ij} = d^2 s,
}
we must in fact have $s=t$.  Multiplying both sides of
\ea{
\sum_j T_{ij} = t
}
by $T_{ik}$ and summing over $i$ we find
\be
1= t \sum_i T_{ik} = t^2.
\ee
So $t^2 = \pm 1$.  If $t=-1$ we can make the replacement $T\to -T$ without changing Eq.~(\ref{eq:equivRepEq}).  We may therefore assume, without loss of generality,
\be
\sum_j T_{ij} =\sum_j T_{ji} = 1
\ee
for all $i$.  It follows that, if we define
\ea{
\tilde{\Pi}_i &= \sum_{j} T_{ij} \Pi_j,
}
then the $\tilde{\Pi}_i$ are also a quasi-SIC.  Moreover
\ea{
U\tilde{\Pi}_i U^{\dagger} &=
  \sum_{j,k} T\vpu{U}_{ij} S^U_{jk} T^{\rm{T}}_{kl} \tilde{\Pi}_l
\nonumber
\\
&= \sum_{l} R^U_{il} \tilde{\Pi}_l.
}
Suppose we now use the $\tilde{\Pi}_i$ to map $\qp$ into operator space by defining
\be
\rho_q = \sum_j \left((d+1) q(j) -\frac{1}{d}\right) \tilde{\Pi}_j
\ee
for all $q\in \qp$.  It follows from the foregoing that, for all $q\in \qp$
\ea{
U\rho_q U^{\dagger} &= \rho_{q'},
}
where
\ea{
q'(j) =  \sum_k R^{U}_{kj}q(k)
}
is also in $\qp$.  It follows that the quasi-state space $\{\rho_q \colon q \in \qp\}$ is invariant under unitary transformations.

\section{An Alternate Route to the Fundamental Inequalities}

In the main text, we began with the urgleichung and eventually arrived at the fundamental inequalities
\begin{equation}
\frac{1}{d(d+1)} \leq \inprod{p}{s} \leq \frac{2}{d(d+1)},
\end{equation}
proving in Theorem~\ref{tm:qplexPolarity} that a self-polar subset of the out-ball $\bo$ is a maximal germ.  Because Theorem~\ref{tm:qplexPolarity} is an if-and-only-if result, it is natural to wonder if one could argue for the fundamental inequalities from a different premise, in which case self-polarity would be a consequence of assuming maximality.

One counterintuitive feature of quantum theory is that two quantum states can be perfectly distinguishable by a von Neumann measurement, yet less distinguishable by an informationally complete measurement~\cite{CFS2002, stacey-qutrit, stacey-hoggar}.  This runs counter to experience with classical probability and stochastic processes, which leads one to think of a non-IC measurement as a coarse-graining (or a convolution by some kernel) of an IC measurement.  If hypothesis $A$ is that the system is in region $A$ of phase space, and hypothesis $B$ is that the system is in region $B$, classical intuition says that hypothesis $A$ and $B$ being perfectly distinguishable means that their regions have no overlap.  Therefore, if we measure where the system is in phase space---the fundamental classical image of what an IC experiment can be---then some outcomes would be consistent with hypothesis $A$, some with hypothesis $B$, and none with both.

In quantum physics, two pure states being orthogonal means that the overlap of their SIC representations is minimal, but minimal is not zero.  If we regard two orthogonal states $\ket{0}$ and $\ket{1}$ as two hypotheses that Alice can entertain about how a system will behave, then there exists some measurement with the property that no outcome is compatible with both hypotheses.  Whatever the outcome of that experiment, one hypothesis or the other will be excluded~\cite{CFS2002}.  But the two hypotheses $\ket{0}$ and $\ket{1}$ have SIC representations $p_0$ and $p_1$, and $\inprod{p_0}{p_1} = 1/(d(d+1))$.  The measurement that defines the SIC representation, although informationally complete, does not itself automatically exclude either hypothesis, because some possible outcomes of it are consistent with both.

With this motivation, we derive quantum state space in the following way.  We again postulate a Bureau of Standards measurement, but we assume as little as possible about the meshing of probability distributions.  Instead of the urgleichung (\ref{eq:gen-urgleichung}), we merely postulate some functional relation~\cite{qbism-greeks},
\begin{equation}
q(j) = F\left(\{p(i), r(j|i): i = 1,\ldots,N\}\right),
\end{equation}
with the property that state vectors with nonzero overlap are incompatible hypotheses with respect to some measurement.  We assume, then, that the inner product of two state-space vectors is bounded below, and take this as an aspect of quantum strangeness.  Then, we assume that certainty is bounded.  This is less strange, since even classically, we can imagine a constraint that probability distributions can never get too focused.  These two postulates tell us that the inner product of two state vectors lies in the interval $[L, U]$.  Note that $U$, being an upper bound on $\inprod{p}{p}$, has an interpretation as an upper bound on an \emph{index of coincidence,} which is inversely related to the \emph{effective population size}~\cite{Leinster12, stacey-thesis, stacey-qutrit}.  Imagine an urn filled with marbles in $N$ different colors.  We draw a marble at random from the urn, note its color, replace it and draw at random again.  If all colors are equally probable, then the probability of obtaining the same color twice in succession is $1/N$.  More generally, if the colors are weighted by some probability vector $p$, then the probability of obtaining the same color twice---i.e., a ``coincidence'' of colors---is $\inprod{p}{p}$.  So, we can take the reciprocal of this quantity as the effective number of colors present.  Regarding the probability vector $p$ as a hypothesis about a system, the effective population size
\begin{equation}
N_{\rm eff}(p) = \frac{1}{\inprod{p}{p}}
\end{equation}
is the effective number of experiment outcomes that are compatible with that hypothesis.  Given two probability vectors $p$ and $s$, we can take
\begin{equation}
N_{\rm eff}(p,s) = N_{\rm eff}(p) N_{\rm eff}(s)\,\inprod{p}{s} = \frac{\inprod{p}{s}}{p^2 s^2}
\end{equation}
as the effective number of outcomes compatible with both hypotheses $p$ and $s$.

By following the logic in Section~\ref{sec:basic}, we can get an upper bound on the size of a Mutually Maximally Distant set.  If we postulate that this bound is saturated, we can relate $L$, $U$ and $N$ to the effective dimensionality:
\begin{equation}
d = 1 + \frac{U - 1/N}{1/N - L}.
\end{equation}
If we take $L = 0$ in the above expression, which we can heuristically
regard as going to the ``classical limit,'' then we end up with $d = NU$.  This
says that the total number of MMD states is the total size of the
sample space ($N$), divided by the area per state, i.e., the effective
population size $1/U$.

Instead of taking $L = 0$, if we choose---for whatever reason---that $L = U/2$ and that $N = d^2$, we get the familiar upper and lower bounds that define a germ.  Postulating that our state space is maximal then implies that it is self-polar.  Because the state space is contained within the probability simplex, it contains the polar of the probability simplex, which is the basis simplex.  By Theorem~\ref{tm:preservation-symmetry}, all the isometries of this set are specified by the regular simplices whose vertices are valid states lying on the out-sphere.

Suppose that $p$ and $s$ are two \emph{pure} states.  Then
\begin{equation}
N_{\rm eff}(p) = N_{\rm eff}(s) = \frac{d(d+1)}{2},
\end{equation}
and
\begin{equation}
N_{\rm eff}(p,s) = \frac{d^2(d+1)^2}{4} \inprod{p}{s} \geq \frac{d(d+1)}{4}.
\end{equation}
Thus, the fundamental inequalities imply that two hypotheses of maximal certainty can only disagree by so much that their overlap is \emph{half} the effective number of outcomes consistent with either hypothesis alone.

We note that Wootters~\cite{Wootters86}, Hardy~\cite{Hardy01} and others~\cite{CFS2001} have used various premises to argue for a relation of the form $N = d^2$.  It bears something of the flavor of a classical state space whose points are labeled by discretized position and momentum~\cite{spekkens2014, Weyl27}.  (And this resonates sympathetically with the fact that the Weyl--Heisenberg group, which is projectively equivalent to $\mathbb{Z}_d \times \mathbb{Z}_d$, is the canonical way to generate SICs~\cite{Zhu10, appleby2016}.) However, at the moment we find it neither an obvious choice nor a consequence of a uniquely compelling assumption.


\begin{thebibliography}{99}
\bibitem{CompactGroup} Given any $M \in \Ot(n)$, columns of $M$ form an orthonormal basis, meaning that any column of any $M$ is a unit vector.  $M$ therefore lies within a topology that is the product of $n$ copies of the sphere $S^{n-1}$.  Because the topology of $\Ot(n)$ is compact, it is a compact group.
\end{thebibliography}
\end{document}